\documentclass[preprint,12pt]{elsarticle}
\usepackage{graphicx}
\usepackage{amsmath}
\usepackage{amsfonts}
\usepackage{amssymb}
\usepackage{mathabx}
\usepackage{algorithm}
\usepackage[noend]{algpseudocode}
\usepackage[title]{appendix}
\usepackage{booktabs}
\usepackage{multirow}
\usepackage{xpatch} 
\usepackage{arydshln} 
\usepackage[update]{epstopdf}
\usepackage[process=all,crop=pdfcrop,cleanup={.dvi,.ps,.pdf,.log,.aux,.bbl}]{pstool}
\usepackage{relsize}
\usepackage{url}
\usepackage{xcolor}
\definecolor{webred}{rgb}{0.5,0,0}
\definecolor{webblue}{rgb}{0,0,0.8}
\usepackage[colorlinks,citecolor=webblue,linkcolor=webred,]{hyperref}
\usepackage{pifont}
\usepackage{siunitx}
\usepackage{mathrsfs}
\usepackage{addfont}
\addfont{OT1}{rsfs10}{\rsfs}

\newcommand{\R}{\mathbb{R}}
\newcommand{\Rn}{\mathbb{R}^n}

\newcommand{\mathbfcal}[1]{\boldsymbol{\mathcal{#1}}}
\newcommand{\gens}{\gamma}
\newcommand{\poly}{s}

\newcommand{\pred}{\varrho}
\newcommand{\fac}{\mathcal{K}}
\newcommand{\concat}{;}
\newcommand{\finally}{\lozenge}
\newcommand{\globally}{\Box}
\newcommand{\until}{U}
\newcommand{\nextOp}{\bigcirc}
\newcommand{\release}{R}
\newcommand{\all}{A}
\newcommand{\initSet}{\mathcal{X}_0}

\algblockdefx{MRepeat}{EndRepeat}{\textbf{repeat}}{}
\algnotext{EndRepeat}
\algblockdefx{MFor}{EndFor}{\textbf{for}}{\textbf{do}}
\algnotext{EndFor}

\makeatletter
\let\OldStatex\Statex
\renewcommand{\Statex}[1][3]{%
  \setlength\@tempdima{\algorithmicindent}%
  \OldStatex\hskip\dimexpr#1\@tempdima\relax}
\makeatother

\usepackage{etoolbox}
\makeatletter
\patchcmd{\ALG@doentity}{\item[]\nointerlineskip}{}{}{}
\makeatother

\newcommand{\cmark}{\ding{51}}%
\newcommand{\xmark}{\ding{55}}%

\allowdisplaybreaks

\makeatletter
\xpatchcmd{\algorithmic}{\itemsep\z@}{\itemsep=0.5ex plus1pt}{}{}
\makeatother

\newtheorem{definition}{Definition}
\newtheorem{proposition}{Proposition}
\newtheorem{theorem}{Theorem}
\newproof{proof}{Proof}

\begin{document}

\begin{frontmatter}
\title{Fully Automated Verification of Linear Time-Invariant Systems against Signal Temporal Logic Specifications via Reachability Analysis}

\journal{Nonlinear Analysis: Hybrid Systems}

\author[SBU]{Niklas Kochdumper} 
\author[SBU]{Stanley Bak}

\affiliation[SBU]{organization={Stony Brook University},
			      city={Stony Brook (NY)},
				  country={USA}}

\begin{abstract}
While reachability analysis is one of the most promising approaches for formal verification of dynamic systems, a major disadvantage preventing a more widespread application is the requirement to manually tune algorithm parameters such as the time step size. Manual tuning is especially problematic if one aims to verify that the system satisfies complicated specifications described by signal temporal logic formulas since the effect the tightness of the reachable set has on the satisfaction of the specification is often non-trivial to see for humans. We address this problem with a fully-automated verifier for linear systems, which automatically refines all parameters for reachability analysis until it can either prove or disprove that the system satisfies a signal temporal logic formula for all initial states and all uncertain inputs. Our verifier combines reachset temporal logic with dependency preservation to obtain a model checking approach whose over-approximation error converges to zero for adequately tuned parameters. 
While we in this work focus on linear systems for simplicity, the general concept we present can equivalently be applied for nonlinear and hybrid systems.

\end{abstract}


\end{frontmatter}

\section{Introduction}

Temporal logic enjoys an increasing popularity in science and engineering, where it is for example used to specify desired behaviors for robots \cite{Plaku2016} and power systems \cite{Xu2017b},
or to formalize traffic rules for road \cite{Maierhofer2020} and marine \cite{Krasowski2021} traffic. However, while the expressiveness of temporal logic is on the one hand advantageous since it allows to model complex behaviors, it unfortunately also makes it very hard to check if a system satisfies a temporal logic formula.
Consequently, while automated verification of linear systems is already possible for simple specifications given by unsafe sets \cite{Wetzlinger2022}, such an approach does not yet exist for the more challenging case of temporal logic specifications.
In this work we address this shortcoming with an automated verifier for linear systems, which decides whether or not the system satisfies a signal temporal logic (STL) \cite{Maler2004} formula for all initial states and uncertain inputs.


\subsection{State of the Art}

 
 

Approaches that check if a system satisfies a temporal logic specification can be divided into the two groups runtime verification and static verification. Runtime verification, which is often realized via monitors \cite{Donze2013,Maler2004}, treats the system as a black box and checks if the observations obtained from the system satisfy the specification. Static verification, on the other hand, considers the case where a model of the system is available and checks if all executions of the model satisfy the temporal logic specification. Since our method falls into the static verification category, we focus on this group from now on.
A standard approach for formal verification against temporal logic specifications is to convert the temporal logic formula into an equivalent acceptance automaton \cite[Sec.~2.7]{Fisher2011}. 
If the system itself can be represented by a finite state automaton, formal verification reduces to checking if there exists an accepting trace for the automaton obtained by taking the automaton product of the system and the acceptance automaton for the negated temporal logic formula \cite[Sec.~5.2]{Baier2008}, which can be realized with standard automaton analysis tools \cite{Gaiser2009}.

\begin{figure}[!tb] 
	\centering
	\includegraphics[width = 0.98\textwidth]{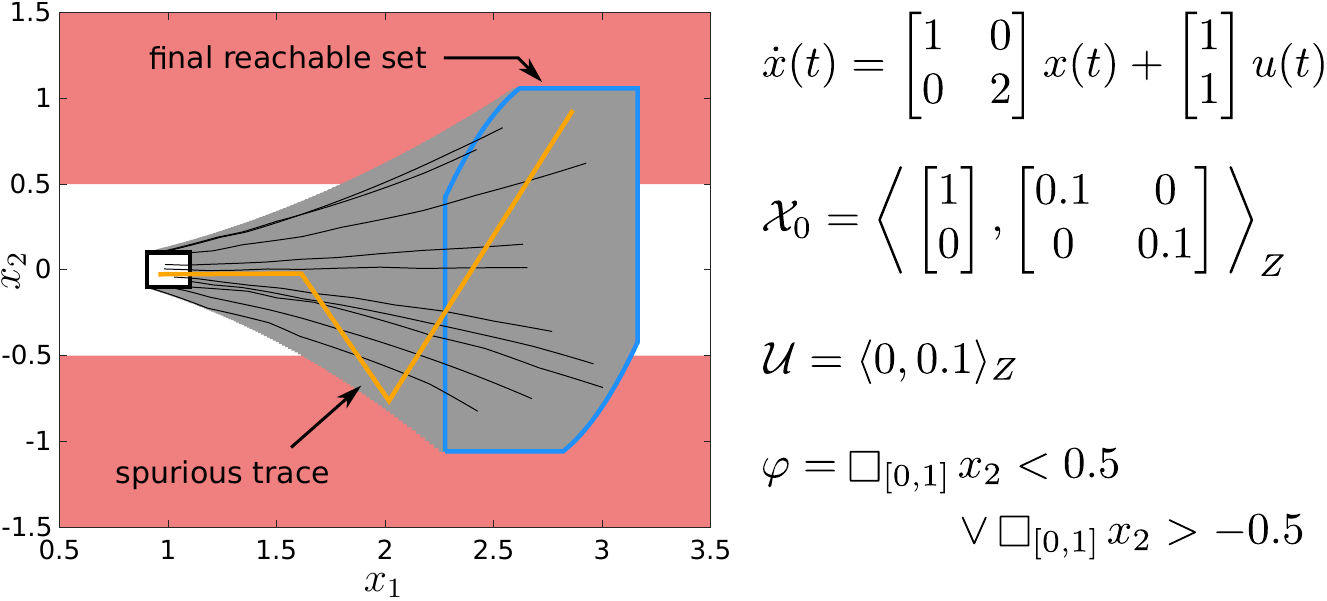}
	\caption{Exemplary verification problem, where the initial set is shown in white with a black border, the reachable set is depicted in gray, the unsafe sets defined by the STL formula $\varphi$ are visualized in red, and exemplary trajectories of the system are shown in black. Even though the system satisfies the temporal logic formula $\varphi$, the reachset temporal logic approach \cite{Roehm2016b} classifies the system as unsafe since it checks if the whole reachable set satisfies the predicates $x_2 < 0.5$ and $x_2 > -0.5$. This procedure produces so-called spurious traces (orange), which are not consistent with the system dynamics.}
	\label{fig:example}
\end{figure}

This concept carries over to dynamic systems where the behavior is described by a differential equation. Many approaches \cite{Tabuada2003,Yordanov2013,Gao2021} focus on discrete-time systems since they can be transformed into an equivalent transition system: By adequately choosing the control input, discrete-time linear systems can be represented by a finite state transition system \cite{Tabuada2003}, and therefore be verified with the same methods as finite state automata. Another approach \cite{Yordanov2013} considers piecewise linear discrete-time systems and uses reachability analysis to refine a finite state automaton abstraction of the system. Reachability analysis is also used by a method for general discrete-time systems \cite{Gao2021}, which transforms the temporal logic specification into a temporal logic tree rather than an acceptance automaton. 

Also methods for continuous-time systems apply the acceptance automaton concept \cite{Bresolin2013,Frehse2018}. They consider temporal logic specifications given in HyLTL \cite{Bresolin2013}, an extension of linear temporal logic \cite{Pnueli1977} to hybrid systems, or as pattern templates \cite{Frehse2018}, which are close to natural language and therefore very intuitive for humans. One major disadvantage of these methods is that reachability analysis for hybrid systems, which are obtained from the automaton product of the continuous dynamics with the acceptance automaton, is quite challenging and often provides very conservative results. Apart from the acceptance automaton concept, there also exist other strategies \cite{Lamport1993,Chen2020,Yu2022,Lee2021,Bae2019}: An early approach \cite{Lamport1993} introduced TLA$^+$, a special logic for describing the behavior of hybrid systems. This enables the formal verification using a theorem prover, which however requires an analytical solution for the differential equation and is therefore only applicable for very simple systems. Other approaches \cite{Yu2022,Lee2021,Bae2019} apply SMT solvers to verify hybrid systems against STL specifications. However, this technique is restricted to systems without inputs and suffers from the curse-of-dimensionality since SMT solvers split the state space. Yet another method \cite{Chen2020} exploits connections between temporal logic operators and reachability problems to compute an approximation for the set of states from which an STL formula can be satisfied based on the Hamilton-Jacobi framework \cite{Mitchell2005}. Finally, it is also possible to convert an STL formula to reachset temporal logic (RTL) \cite{Roehm2016b}, a special logic that can be directly evaluated on the reachable set. All existing approaches have the disadvantage that parameters such as the time step size have to be tuned manually by the user for the verification to succeed.

\subsection{Contribution}

In this work we present the first fully-automated verifier for linear systems and signal temporal logic specifications, which is based on the reachset temporal logic approach \cite{Roehm2016b}. One major drawback of the previous RTL method is demonstrated by the example shown in Fig.~\ref{fig:example}: RTL checks if the atomic predicates are satisfied by the whole reachable set, which yields very conservative results, especially if the reachable set becomes large. We solve this problem by keeping track which parts of the initial set and input set satisfy or violate the predicates at specific points in time, as it is visualized in Fig.~\ref{fig:temporalLogic}. In summary, our verifier has the following features:
\begin{itemize}
	\item[$\bullet$] The algorithm is guaranteed to converge to the solution in finite time for decidable problem instances.
	\item[$\bullet$] Our verifier is fully-automated, so no manual parameter tuning is required.
	\item[$\bullet$] In case the system violates the specification we return a counterexample in form of a falsifying trajectory, which might provide valuable insights for a system engineer.
	\item[$\bullet$] Since our verifier naturally divides the initial set and input set into parts that satisfy or violate the temporal logic formula, it can also be used to determine a safe set of states or cut away all states that yield a violation, which is very useful for controller synthesis and set-based prediction.
\end{itemize}
We will demonstrate all of these features on several interesting benchmarks.

\subsection{Overview}

The remainder of the paper is structured as follows: 
We first introduce some preliminaries including notations, set representations as well as operations, reachable sets, and temporal logic in Sec.~\ref{sec:prelim}.
Afterward, we specify the formal problem definition in Sec.~\ref{sec:problem}. To keep track of dependencies between reachable states and initial states as well as inputs, we require a special type of reachability analysis, which is presented in Sec.~\ref{sec:reach}. Next, we introduce our model checking approach in Sec.~\ref{sec:modelChecking}, before we describe the overall verifier in Sec.~\ref{sec:verify}. Finally, in Sec.~\ref{sec:numEx}, we demonstrate the performance of our automated verifier on several challenging benchmarks, before we provide an outlook to future directions in Sec.~\ref{sec:future}.

\section{Preliminaries}
\label{sec:prelim}

Let us first introduce the notation as well as some preliminaries and definitions.

\subsection{Notation}

Scalars and vectors are denoted by lowercase letters, whereas matrices are denoted by uppercase letters.
Given a vector $v \in \R^{n}$, $v_{(i)}$ represents the $i$-th entry and $\| v\|_p$ is the p-norm. Moreover, all vectors considered in this work are column vectors.
For a matrix $M \in \R^{w \times n}$, $M_{(i,\cdot)}$ refers to the $i$-th row and $M_{(\cdot,j)}$ to the $j$-th column.
The identity matrix of dimension~$n$ is denoted by $I_n$ and we use $\mathbf{0}$ and $\mathbf{1}$ to represent vectors and matrices of proper dimension containing only zeros or ones. Moreover, we denote the horizontal concatenation of two matrices $M_1, M_2$ by $[M_1~M_2]$ and the vertical concatenation by $[M_1 \concat M_2]$.
The floor function $\lfloor x \rfloor$ with $x \in \R$ returns the largest integer that is smaller than $x$, and $x~\text{mod}~y$ with $x,y \in \R$ denotes the modulo operator.
While sets are denoted by standard calligraphic letters $\mathcal{S}$, interval matrices are denoted by bold calligraphic letters: $\mathbfcal{M} = [\underline{M},\overline{M}] = \{ M \in \R^{w \times n} ~ | ~ \underline{M} \leq M \leq \overline{M} \}$, where the inequality is evaluated element-wise. $||\mathbfcal{M}||_F$ denotes the Frobenius norm of an interval matrix.
Intervals are a special case of interval matrices, where the lower and upper bounds are vectors.
Tuples are denoted by bold uppercase letters: Given a $n$-tuple $\mathbf{H} = (h_1,\dots,h_n)$, $|\mathbf{H}| = n$ denotes the cardinality of the tuple and $\mathbf{H}_{(i)} = h_i$ refers to the $i$-th entry of tuple $\mathbf{H}$. Moreover, given two tuples $\mathbf{H} = (h_1, \dots ,h_n)$ and $\mathbf{K} = (k_1, \dots, k_w)$, $( \mathbf{H}, \mathbf{K} ) = (h_1, \dots ,h_n,k_1,\dots,k_w)$ denotes their concatenation, operation $\mathbf{H} \setminus \mathbf{K}$ removes all elements in $\mathbf{K}$ from $\mathbf{H}$, and the empty tuple is denoted by $\emptyset$. Finally, given a matrix $M \in \R^{w \times n}$ and a tuple $\mathbf{H} = (h_1,\dots,h_m)$ with $h_1,\dots,h_m \in \mathbb{N}_{\leq n}$, we use the shorthand notation $M_{(\cdot,\mathbf{H})} = [M_{(\cdot,h_1)} ~ \dots ~ M_{(\cdot,h_m)}]$.

\subsection{Set Representations and Operations}

Given sets $\mathcal{S}_1,\mathcal{S}_2 \subset \Rn{}$ and a matrix $M \in \R^{w \times n}$, we require the set operations linear map $M \mathcal{S}_{1}$, Minkowski sum $\mathcal{S}_{1} \oplus \mathcal{S}_{2}$, Cartesian product $\mathcal{S}_1 \times \mathcal{S}_2$, intersection $\mathcal{S}_1 \cap \mathcal{S}_2$, and convex hull $conv(\mathcal{S}_{1},\mathcal{S}_{2})$, which are defined as 
\begin{align}
	& M \mathcal{S}_{1} := \{ M s ~|~ s \in \mathcal{S}_{1} \} , \label{eq:defLinTrans} \\
	& \mathcal{S}_{1} \oplus \mathcal{S}_{2} := \{ s_1 + s_2 ~|~ s_1 \in \mathcal{S}_{1}, s_2 \in \mathcal{S}_{2} \} , \label{eq:defMinSum} \\
	& \mathcal{S}_{1} \times \mathcal{S}_{2} := \{ [s_1 \concat s_2] ~|~ s_1 \in \mathcal{S}_{1}, s_2 \in \mathcal{S}_{2} \} , \label{eq:defCartProd} \\
	& \mathcal{S}_1 \cap \mathcal{S}_2 := \{s ~|~ s \in \mathcal{S}_1 \wedge s \in \mathcal{S}_2 \}, \label{eq:defIntersection} \\
\begin{split} \label{eq:defConvHull}
	& conv(\mathcal{S}_{1},\mathcal{S}_{2}) := \bigg \{\sum_{i=1}^{n+1} \lambda_i\,s_i~\bigg|~s_i \in \mathcal{S}_1 \cup \mathcal{S}_2,~\lambda_i \geq 0,~\sum_{i=1}^{n+1} \lambda_i = 1 \bigg\}.
\end{split}
\end{align}
In this paper, we represent reachable sets by zonotopes:

\begin{definition} \label{def:zonotope}
(Zonotope) Given a center vector $c \in \Rn{}$ and a generator matrix $G \in \R^{n \times \gens{}}$, a zonotope $\mathcal{Z} \subset \Rn{}$ is
\begin{equation*}
	\mathcal{Z} := \bigg\{ c + \sum_{i = 1}^{\gens} G_{(\cdot,i)} \, \alpha_i ~\bigg| ~ \alpha_i \in [-1,1] \bigg\},
\end{equation*}
where the scalars $\alpha_i$ are called factors. We use the shorthand $\mathcal{Z} = \langle c,G \rangle_Z$.
\end{definition}
For zonotopes $\mathcal{Z}_{1} = \langle c_1,G_1 \rangle_Z,\mathcal{Z}_{2} = \langle c_2,G_2 \rangle_Z \subset \Rn{}$ with $\gamma_1$ and $\gamma_2$ generators, respectively, linear map, Minkowski sum, and convex hull can be computed as \cite[Eq.~(2.1)-(2.2)]{Althoff2010a}
\begin{align}
	& M \mathcal{Z}_{1} = \langle M c_1, M G_1 \rangle_Z , \label{eq:zonoLinTrans} \\
	& \mathcal{Z}_{1} \oplus \mathcal{Z}_{2} = \langle c_1 + c_2, [G_1~G_2] \rangle_Z , \label{eq:zonoMinSum} \\
	& conv(\mathcal{Z}_{1},\mathcal{Z}_{2}) \subseteq \big \langle 0.5(c_1 + c_2),~\big[ 0.5 (G_1 + G^{(1)}_2) \nonumber \\
	& \qquad \qquad \qquad \qquad  0.5 (G_1 - G_2^{(1)})~ 0.5(c_1 - c_2)~ G_2^{(2)}  \big] \big \rangle_Z \label{eq:zonoConvHull}
\end{align} 
with
\begin{equation*}
	G_2^{(1)} = [G_{2(\cdot,1)} \, \dots \, G_{2(\cdot,\gens{}_1)}], ~ G_2^{(2)} = [G_{2(\cdot,\gens{}_1+1)} \, \dots \, G_{2(\cdot,\gens{}_2)}],
\end{equation*}
where we assume without loss of generality that $\mathcal{Z}_{2}$ has more generators than $\mathcal{Z}_{1}$. In addition, the multiplication $\mathbfcal{I} \, \mathcal{Z}$ of an interval matrix $\mathbfcal{I}$ with a zonotope $\mathcal{Z}$ can be enclosed as specified in \cite[Thm.~4]{Althoff2010a}, and the operation $\texttt{interval}(\mathcal{Z})$ returns
the interval enclosure according to \cite[Prop. 2.2]{Althoff2010a}.

Another set representation we require are polytopes, for which we consider the halfspace representation:

\begin{definition}
	(Polytope) Given a constraint matrix $C \in \R^{\poly \times n}$ and a constraint offset $d \in \R^{\poly}$, a polytope $\mathcal{P} \subseteq \R^{n}$ is defined as
	\begin{equation*}
		\mathcal{P} := \big \{ x \in \Rn ~\big|~ C \,x \leq d \big \}.
	\end{equation*}
	We use the shorthand $\mathcal{P} = \langle C,d \rangle_P$.
	\label{def:polytope}
\end{definition}
The intersection of two polytopes $\mathcal{P}_1 = \langle C_1,d_1\rangle_P,\mathcal{P}_2 = \langle C_2,d_2 \rangle_P \subseteq \Rn$ can be computed as
\begin{equation}
	\mathcal{P}_1 \cap \mathcal{P}_2 = \big \langle [C_1 \concat C_2],[d_1 \concat d_2] \big \rangle_P.
\end{equation}
Computation of the intersection might result in redundant halfspaces, which can be removed using linear programming.

\subsection{Reachability Analysis}

The reachable set of a dynamic system is defined as follows:

\begin{definition}
	(Reachable Set) We consider a dynamic system whose behavior is described by the differential equation
	\begin{equation}
		\dot x(t) = f\big(x(t),u(t)\big),
		\label{eq:ode}
	\end{equation}
	where $x(t) \in \Rn$ is the system state and $u(t) \in \R^m$ is the input. Given an initial set $\initSet \subset \Rn$ and an input set $\mathcal{U} \subset \R^m$, the reachable set at time $t \geq 0$ is defined as
	\begin{equation*}
		\mathcal{R}^{\text{\normalfont e}}(t) := \big \{ \xi(t,x_0,u(\cdot))~\big |~ x_0 \in \initSet, \forall \theta \in [0,t]:~ u(\theta) \in \mathcal{U} \big \},
	\end{equation*}
	where $\xi(t,x_0,u(\cdot))$ denotes the solution to \eqref{eq:ode} for the initial state $x_0 = x(0)$ and the input signal $u(\cdot)$. 
	\label{def:reachSet}
\end{definition}
Since the exact reachable set $\mathcal{R}^{\text{e}}(t)$ as defined in Def.~\ref{def:reachSet} cannot be computed in general, the goal of reachability analysis is to calculate a tight enclosure $\mathcal{R}(t) \supseteq \mathcal{R}^{\text{\normalfont e}}(t)$ instead. Moreover, it is common practice to compute the reachable set for consecutive time intervals $\tau_i = [t_i,t_{i+1}]$ with $t_i = i \cdot \Delta t$, $i \in \{0,\dots,t_{\text{end}}/\Delta t\}$, where $\Delta t$ is the time step size and $t_{\text{end}}$ is the final time. We assume without loss of generality that time starts at $t_0 = 0$ and $t_{\text{end}}$ is a multiple of $\Delta t$. Consequently, the reachable set for the whole time horizon $\mathcal{R}([0,t_{\text{end}}])$ is given by a sequence
\begin{equation}
	\mathcal{R}(t_0),\mathcal{R}(\tau_0),\mathcal{R}(t_1),\mathcal{R}(\tau_1),\dots,\mathcal{R}(\tau_{t_{\text{end}}/\Delta t -1}),\mathcal{R}(t_{\text{end}})
	\label{eq:reachSeq}
\end{equation} 
of time point reachable sets $\mathcal{R}(t_i)$ and time interval reachable sets $\mathcal{R}(\tau_i)$. For the remainder of the paper, we use the shorthand notation $\forall u(\cdot) \in \mathcal{U}$ to refer to the set $\{u(\cdot)\,|\, \forall \theta \in [0,t_{\text{end}}]:\,u(\theta) \in \mathcal{U} \}$ of all input signals that are contained in the input set $\mathcal{U}$ at all times.

Some reachability algorithms preserve dependencies between initial states and inputs and the corresponding reachable states \cite{Kochdumper2020c}. For those algorithms, the computed enclosure of the reachable set can be used to construct a symbolic function that approximates the solution to the differential equation within a guaranteed error bound:
\begin{definition}(Dependency Preservation)
	Let $\mathcal{R}(t)$ be an enclosure of the reachable set for the initial set $\mathcal{X}_0$ and the set of uncertain inputs $\mathcal{U}$ computed with a reachability algorithm $\mathcal{A}$. The reachability algorithm $\mathcal{A}$ is dependency preserving if the computed reachable set $\mathcal{R}(t)$ can be used to construct a function $\mu(t,x_0,u(\cdot))$ that approximates the solution $\xi(t,x_0,u(\cdot))$ of the differential equation in \eqref{eq:ode} within some time-varying error bound $\mathcal{E}(t)$ for all initial states $x_0 \in \mathcal{X}_0$ and input signals $u(\cdot) \in \mathcal{U}$:
	\begin{equation*}
		\forall t \in [0,t_{\text{\normalfont end}}], \forall x_0 \in \mathcal{X}_0, \forall u(\cdot) \in \mathcal{U}: ~\xi(t,x_0,u(\cdot)) \in \mu(t,x_0,u(\cdot)) \oplus \mathcal{E}(t) \, \subseteq \, \mathcal{R}(t),
	\end{equation*}
	where the trivial solution $\mu(t,x_0,u(\cdot)) = \mathbf{0}$ and $\mathcal{E}(t) = \mathcal{R}(t)$ is only allowed if the reachable set enclosure $\mathcal{R}(t)$ is singleton.
	\label{def:dependency}
\end{definition}
The main advantage of dependency preservation is that with the analytical relation $\mu(t,x_0,u(\cdot))$ between initial states and inputs and the corresponding reachable states enclosures of reachable sets for single initial states and inputs can be computed very efficiently in terms of a simple function evaluation.

\subsection{Temporal Logic}

We consider specifications in signal temporal logic \cite{Maler2004}:

\begin{definition}
	(Signal Temporal Logic) The syntax of a signal temporal logic formula over a finite set of atomic predicates $\pred \in \mathcal{AP}$ is
	\begin{equation*}
		\varphi := \pred \,|\, \neg \varphi \,|\, \varphi_1 \vee \varphi_2 \,|\, \varphi_1 \wedge \varphi_2 \,|\, \varphi_1 \, \until_{[a,b]} \, \varphi_2 \,|\, \varphi_1 \, \release_{[a,b]} \, \varphi_2 \,|\, \finally_{[a,b]} \, \varphi  \,|\, \globally_{[a,b]} \, \varphi \,|\, \nextOp_a \, \varphi
	\end{equation*}
	with $a,b \in \R_{\geq 0}$ and $b \geq a$, where we in addition to the until operator $\varphi_1 \, \until_{[a,b]} \, \varphi_2$ consider the operators release $\varphi_1 \, \release_{[a,b]} \, \varphi_2 := \neg( \neg \varphi_1 \, \until_{[a,b]} \, \neg \varphi_2)$, finally $\finally_{[a,b]} \, \varphi := \mathrm{true}~ \until_{[a,b]} \, \varphi$, globally $\globally_{[a,b]} \, \varphi := \neg \finally_{[a,b]} \, \neg \varphi$, and next $\nextOp_a \, \varphi := \finally_{[a,a]} \, \varphi$. For a trace $\xi(t)$, the semantics of a signal temporal logic formula is defined as follows:
	\begin{equation*}
		\begin{split}
			\xi \vDash \pred & ~~ \Leftrightarrow ~~ \pi_{\pred}(\xi(0)) = \mathrm{true} \\
			\xi \vDash \neg \varphi & ~~ \Leftrightarrow ~~ \neg (\xi \vDash \varphi) \\
			\xi \vDash \varphi_1 \vee \varphi_2 & ~~ \Leftrightarrow ~~ (\xi \vDash \varphi_1) \vee (\xi \vDash \varphi_2) \\
			\xi \vDash \varphi_1 \, \until_{[a,b]} \, \varphi_2 & ~~\Leftrightarrow ~~ \exists t \in [a,b]: \xi_t \vDash \varphi_2 ~\wedge~ \forall t' \in [0,t): \xi_{t'} \vDash \varphi_1
		\end{split}
	\end{equation*}
	using a predicate evaluation function $\pi_{\pred}$ and notation $\xi_a(t) := \xi(t+a)$. The semantics for the conjunction follows from the equality $\varphi_1 \wedge \varphi_2 = \neg(\neg \varphi_1 \vee \neg \varphi_2)$.
\end{definition}
In this work we consider atomic predicates given by linear inequality constraints such as $2 \, x_1 - 3 \, x_2 \leq 1$.
Please note that this also includes polytope containment since $x \in \mathcal{P}$ with $\mathcal{P} = \langle C,d \rangle_P$ can be represented as $C_{(1,\cdot)} x \leq d_{(1)} \wedge \dots \wedge C_{(\poly,\cdot)} x \leq d_{(\poly)}$. Reachset temporal logic \cite[Def.~2]{Roehm2016b} is a special type of logic that can be directly evaluated on a reach sequence \eqref{eq:reachSeq}:

\begin{definition}
	(Reachset Temporal Logic) The syntax of a reachset temporal logic formula\footnote{In contrast to \cite[Def.~2]{Roehm2016b}, we use a shortened definition of RTL that is restricted to the relevant parts.} over a finite set of atomic predicates $\pred \in \mathcal{AP}$ is
	\begin{equation*}
		\varphi := \all \,\pred ~|~ \varphi_1 \vee \varphi_2 ~|~ \varphi_1 \wedge \varphi_2 ~|~  \nextOp_a \, \varphi
	\end{equation*}
	with $a = i \cdot \Delta t/2$ and $i \in \mathbb{N}_{\geq 0}$, where $\Delta t$ is the time step size. The all operator $\all \, \pred$ specifies that a predicate is satisfied for all states inside the reachable set, and the next operator $\nextOp_a \, \varphi$ refers to a time point reachable set if $a$ is an integer multiple of $\Delta t$ and to a time interval reachable set otherwise. For a reachable set $\mathcal{R}(t)$, the semantics of a reachset temporal logic formula is defined as follows:
	\begin{equation*}
		\begin{split}
			\mathcal{R} \vDash \all \, \pred & ~~ \Leftrightarrow ~~ \forall r \in \mathcal{R}(0): ~\pi_{\pred}(r) = \mathrm{true} \\
			\mathcal{R} \vDash \varphi_1 \vee \varphi_2 & ~~ \Leftrightarrow ~~ (\mathcal{R} \vDash \varphi_1) \vee (\mathcal{R} \vDash \varphi_2) \\
			\mathcal{R} \vDash \varphi_1 \wedge \varphi_2 & ~~ \Leftrightarrow ~~ (\mathcal{R} \vDash \varphi_1) \wedge (\mathcal{R} \vDash \varphi_2) \\
			\mathcal{R} \vDash \nextOp_a \, \varphi & ~~\Leftrightarrow ~~ \begin{cases}\mathcal{R}_a \vDash \varphi, & a ~ \text{\normalfont mod} ~ \Delta t = 0\\ \forall t \in \left \lfloor \frac{a}{\Delta t} \right \rfloor + [0, \Delta t]: ~ \mathcal{R}_t \vDash \varphi, & \text{\normalfont otherwise}\end{cases}
		\end{split}
	\end{equation*}
	using a predicate evaluation function $\pi_{\pred}$ and notation $\mathcal{R}_a(t) := \mathcal{R}(t+a)$.
\end{definition}
The paper \cite{Roehm2016b} that introduced RTL also provides an approach for converting a temporal logic formula in STL to RTL. Given an STL formula $\varphi$, the resulting RTL formula $\varphi_{\text{rtl}}$ has according to \cite[Thm.~1]{Roehm2016b} the following structure:
\begin{equation*}
	\varphi_{\text{rtl}} = \bigwedge_{h = 0}^H \bigvee_{j = 0}^J \nextOp_{j \, \Delta t/2} \bigvee_{k=0}^K \all \, \pred_{hjk},
\end{equation*}
where $H$, $J$, and $K$ denote the number of conjunctions and disjunctions. 
Since we consider the case where the atomic predicates $\pred_{hjk}$ are all given by linear inequality constraints, the entailment check $\mathcal{R} \vDash \varphi_{\text{rtl}}$ can be equivalently formulated in terms of intersection checks between the reachable set $\mathcal{R}(t)$ and polytopes $\mathcal{P}_{hjkl}$ in halfspace representation \cite[Sec.~5]{Roehm2016b}:
\begin{equation}
	\mathcal{R} \vDash \varphi_{\text{rtl}} ~~ \Leftrightarrow ~~\bigwedge_{h=0}^H \bigvee_{j=0}^J \bigvee_{k=0}^K \bigg( \widetilde{\mathcal{R}}(j \, \Delta t/2) \cap \bigcup_{l=0}^L \mathcal{P}_{hjkl} = \emptyset \bigg),
	\label{eq:RTLfinal}
\end{equation}
where $L$ is the number of polytopes required to represent the atomic predicates $\pred_{hjk}$ and the auxiliary variable $\widetilde{\mathcal{R}}(t)$ is used to distinguish between time point and time interval reachable sets:
\begin{equation*}
	\widetilde{\mathcal{R}}(t) = \begin{cases} \mathcal{R}(t_j), & t ~ \text{\normalfont mod} ~ \Delta t = 0  \\ \mathcal{R}(\tau_j), & \text{\normalfont otherwise} \end{cases} ~~~~\text{with} ~~~~ j = \lfloor t / \Delta t \rfloor.
\end{equation*}
Since we use the same time step size $\Delta t$ for reachability analysis and for the time-discretization of reachset temporal logic, \eqref{eq:RTLfinal} can be directly evaluated on a reach sequence \eqref{eq:reachSeq}.
An example demonstrating the conversion from STL to RTL is provided in \ref{sec:appendix}.

According to \cite{Roehm2016b}, it holds that if the reachable set satisfies the RTL formula $\varphi_{\text{rtl}}$, then it is guaranteed that all traces contained in the reachable set satisfy the corresponding STL formula $\varphi$. Due to the conservatism introduced by the time-discretization for RTL with $\Delta t$, the contrary does not hold, meaning that satisfaction of the STL formula does not automatically imply satisfaction of the RTL formula. However, since the conservatism introduced by the time-discretization converges to zero for $\Delta t \to 0$ \cite{Roehm2016b}, satisfaction of STL and RTL becomes equivalent in the limit case $\Delta t \to 0$.
%

\section{Problem Formulation}
\label{sec:problem}

We consider linear time-invariant systems
\begin{align}
	\dot{x}(t)	&= A \, x(t) + B \, u(t) \label{eq:linsys}
\end{align}
with $A \in \R^{n \times n}$, $B \in \R^{n \times m}$, where $x(t) \in \R^{n}$ is the state and $u(t) \in \R^{m}$ is the input.
The initial state $x(t_0)$ is uncertain within the initial set $\initSet \subset \R^{n}$ and the input $u(t)$ is uncertain within the input set $\mathcal{U} \subset \R^{m}$. In this work, we assume that $\initSet$ and $\mathcal{U}$ are represented by zonotopes.
While we consider the set of uncertain inputs $\mathcal{U}$ to be constant over time for simplicity, the extension to a time-varying set of uncertain inputs $\mathcal{U}(t)$ is straightforward. Time-varying sets of uncertain inputs can for example be used to pass external signals such as reference trajectories or a sequence of control commands to the model.

Given an STL formula $\varphi$, our goal is to decide whether the system \eqref{eq:linsys} satisfies $\varphi$ for all initial states $x_0 \in \initSet$ and all input signals $u(\cdot) \in \mathcal{U}$: 
\begin{equation*}
	\forall t, \forall x_0 \in \initSet, \forall u(\cdot) \in \mathcal{U}: ~ \xi(t,x_0,u(\cdot)) \vDash \varphi,
\end{equation*}
where $\xi(t,x_0,u(\cdot))$ is the solution to \eqref{eq:linsys} for the initial state $x_0 = x(t_0)$ and the input signal $u(\cdot)$. 

\section{Reachability Analysis}
\label{sec:reach}

For the approach presented in this paper we require a reachability algorithm that is dependency preserving according to Def.~\ref{def:dependency}. To construct such an algorithm, we modify the reachability algorithm in \cite[Sec.~3.2]{Althoff2010a} by using a different enclosure for the reachable set due to uncertain inputs $\mathcal{P}(\Delta t)$. In particular, to preserve dependencies on the inputs, we approximate the reachable set due to time-varying inputs with the reachable set due to constant inputs $\mathcal{P}_c(\Delta t) = T \, \mathcal{U}_0$ in each time step, where 
\begin{equation}
	T = A^{-1} (e^{A \Delta t} - I_n)
	\label{eq:propMat}
\end{equation}
is the propagation matrix for constant inputs and $\mathcal{U}_0 = B \big (\mathcal{U} \oplus (-c_u)\big) \subset \R^{n}$ is an auxiliary variable defined using the geometric center $c_u$ of the input set $\mathcal{U} = \langle c_u,G_u\rangle_Z$. If the matrix $A$ is not invertible, $A^{-1}$ in \eqref{eq:propMat} can be integrated into the power-series for the exponential matrix \cite[Sec.~IV]{Kochdumper2022}. Since $\mathcal{P}_c(\Delta t) = T \, \mathcal{U}_0$ is computed using a linear map, it preserves dependencies according to \cite[Tab.~1]{Kochdumper2020c}. We account for the approximation error with the bloating term
\begin{equation}
	\mathcal{D} = \bigg( \sum_{j=1}^\kappa \frac{A^j \Delta t^{j+1}}{(j+1)!} \bigg) \mathcal{U}_0 \oplus \bigoplus_{j=1}^\kappa \frac{A^j \Delta t^{j+1}}{(j+1)!}\, \mathcal{U}_0 \oplus 2 \,  \Delta t \, \mathbfcal{E} \, \mathcal{U}_0,
	\label{eq:inputDiff}
\end{equation}
which according to \cite[Prop.~2]{Wetzlinger2022} encloses the difference between constant and time-varying inputs. Consequently, we overall obtain the following enclosure for the reachable set due to time-varying inputs:
\begin{equation}
	\mathcal{P}(\Delta t) \subseteq T \, \mathcal{U}_0 \oplus \mathcal{D}.
	\label{eq:inputReach}
\end{equation}
Another essential part of the reachability algorithm in  \cite[Sec.~3.2]{Althoff2010a} is the set $\mathcal{C}_i$ that accounts for the curvature of trajectories and is according to \cite[Sec.~3.2]{Althoff2010a} given as 
\begin{equation} \label{eq:Cu}
	\mathcal{C}_i = (e^{A \Delta t})^{i} \big( \mathbfcal{F} \, \mathcal{X}_0 \oplus \mathbfcal{G} \, \widetilde{u} \big)
\end{equation}
using the auxiliary variable $\widetilde u = B c_u \in \R^{n}$ and the interval matrices \cite[Sec.~3.2]{Althoff2010a}
\begin{equation}
	\mathbfcal{F} = \mathbfcal{T}^{(\kappa)} A \oplus \mathbfcal{E}, ~ \mathbfcal{G} = \mathbfcal{T}^{(\kappa + 1)} \oplus \mathbfcal{E} \Delta t, ~ \mathbfcal{T}^{(o)} = \bigoplus_{j=2}^o \big[\big( j^{\frac{-j}{j-1}} - j^{\frac{-1}{j-1}} \big) \Delta t^j,0\big] \frac{A^{j-1}}{j!},
	\label{eq:curv}
\end{equation}
%
where the parameter $\kappa$ and the interval matrix $\mathbfcal{E}$ are the truncation order and the remainder of the Taylor series for the exponential matrix \cite[Eq.~(3.2)]{Althoff2010a}:
\begin{equation*} 
	\mathbfcal{E} = [-E,E] , \quad
	E = e^{|A|\Delta t} - \sum_{j=0}^{\kappa} \frac{\big( |A|\Delta t \big)^j}{j!} .
\end{equation*}
Using the enclosure of the reachable set due to time-varying inputs $\mathcal{P}(\Delta t)$ in \eqref{eq:inputReach} as well as the curvature enclosure $\mathcal{C}_i$ in \eqref{eq:Cu}, the reachability algorithm in \cite[Sec.~3.2]{Althoff2010a} finally computes tight enclosures of the time point reachable set $\mathcal{R}(t_i)$ and the time interval reachable set $\mathcal{R}(\tau_i)$ as follows: 
\begin{equation}
\begin{split}
	& \mathcal{R}(t_{i}) = \mathcal{H}(t_{i}) \oplus \mathcal{P}(t_{i}), \\
	& \mathcal{R}(\tau_{i}) = conv\big(\mathcal{H}(t_i), \mathcal{H}(t_{i+1})\big) \oplus \mathcal{P}(t_{i+1}) \oplus \mathcal{C}_i,
\end{split}
\label{eq:reach}
\end{equation}
where the homogeneous solution $\mathcal{H}(t_i)$ resulting from the propagation of the initial set and constant inputs as well as the particular solution $\mathcal{P}(t_i)$ due to uncertain time-varying inputs are computed using the following propagation scheme:
\begin{equation}
\begin{split}
	\mathcal{H}(t_{i+1}) &= e^{A \Delta t} \, \mathcal{H}(t_{i}) \oplus T \, \widetilde{u}, \\
	\mathcal{P}(t_{i+1}) &= e^{A \Delta t} \, \mathcal{P}(t_i) \oplus \mathcal{P}(\Delta t),
\end{split}
\label{eq:propScheme}
\end{equation} 
where the initial values are $\mathcal{H}(0) = \mathcal{X}_0$ and $\mathcal{P}(0) = \mathbf{0}$. The resulting reachability algorithm is summarized in Alg.~\ref{alg:reach}, where we enclose $\mathcal{C}_i$ and $\mathcal{D}$ by intervals to reduce the number of generators of the reachable set.

\begin{algorithm}[!tb]
	\caption{Reachability Analysis} \label{alg:reach}
	\vspace{2pt}
	{\raggedright \textbf{Input:} Linear system $\dot x = A x + B u$, initial set $\initSet = \langle c_x,G_x\rangle_Z$, input set $\mathcal{U}= \langle c_u,G_u \rangle_Z$, final time $t_{\text{end}}$, time step size $\Delta t$, truncation order $\kappa$.
	
	\textbf{Output:} Reach sequence $\mathcal{R}(t_0),\mathcal{R}(\tau_0),\mathcal{R}(t_1),\mathcal{R}(\tau_1),\dots,\mathcal{R}(\tau_{t_{\text{end}}/\Delta t-1}),\mathcal{R}(t_{\text{end}})$ with time point reachable sets $\mathcal{R}(t_i)$ and time interval reachable sets $\mathcal{R}(\tau_i)$.
	}
	\begin{algorithmic}[1]
		\State $\widetilde{u} \gets B \, c_u$, $\mathcal{U}_0 \gets B \big( \mathcal{U} \oplus (-c_u) \big) $, $T \gets \eqref{eq:propMat}$, $\mathbfcal{F},\mathbfcal{G} \gets \eqref{eq:curv}$, $\mathcal{D} \gets \eqref{eq:inputDiff}$
		\State $t_0 \gets 0$, $\mathcal{H}(t_0),\mathcal{R}(t_0) \gets \initSet$, $\mathcal{P}(t_0),\mathcal{P}_c(t_0),\mathcal{D}_0 \gets \mathbf{0}$
		\State $\mathcal{C}_0 \gets \texttt{interval}(\mathbfcal{F} \, \mathcal{X}_0 \oplus \mathbfcal{G} \, \widetilde{u})$
		\For{$i \gets 0$ to $t_{\text{end}}/\Delta t -1$}
			\State $t_{i+1} \gets t_i + \Delta t$, $\tau_i \gets [t_i,t_{i+1}]$, $\mathcal{C}_{i+1} \gets e^{A \Delta t} \mathcal{C}_i$  
			\State $\mathcal{H}(t_{i+1}) \gets e^{A \Delta t} \, \mathcal{H}(t_{i}) \oplus T \, \widetilde{u}$, $\mathcal{D}_{i+1} \gets \texttt{interval}(e^{A \Delta t} \, \mathcal{D}_i \oplus \mathcal{D})$
			\State $\mathcal{P}_c(t_{i+1}) \gets e^{A \Delta t} \, \mathcal{P}_c(t_i) \oplus T \, \mathcal{U}_0$, $\mathcal{P}(t_{i+1}) \gets \mathcal{P}_c(t_{i+1}) \oplus \mathcal{D}_{i+1}$
			\State $\mathcal{R}(t_{i+1}) \gets \mathcal{H}(t_{i+1}) \oplus \mathcal{P}(t_{i+1})$
			\State $\mathcal{R}(\tau_{i}) \gets conv\big(\mathcal{H}(t_i), \mathcal{H}(t_{i+1})\big) \oplus \mathcal{P}(t_{i+1}) \oplus \mathcal{C}_i$
		\EndFor
	\end{algorithmic}
\end{algorithm}


An important property of Alg.~\ref{alg:reach} is that it is dependency preserving for piecewise constant inputs. Please note that it is sufficient to consider piecewise constant inputs in our case since those approximate time-varying inputs arbitrary well for $\Delta t \to 0$. We therefore now show how to construct the approximate solution $\mu(t,x_0,u(\cdot))$ to the differential equation as well as the error bound $\mathcal{E}(t)$ in Def.~\ref{def:dependency} from the computed reachable set:

\begin{proposition}
	We consider the initial set $\initSet = \langle c_x,G_x \rangle_Z \subset \Rn$ with $\gens_x$ generators, the input set $\mathcal{U} = \langle c_u,G_u \rangle_Z \subset \R^m$ with $\gens_u$ generators, and the corresponding enclosure of the reachable set $\mathcal{R}(t)$ computed with Alg.~\ref{alg:reach} using time step size $\Delta t$. Given an initial state $x_0 \in \initSet$ as well as the piecewise constant input signal $u(\cdot) \in \mathcal{U}$ defined as
	\begin{equation}
		x_0 = c_x + G_x \alpha_x,\quad \forall t \in [(j-1) \Delta t,j \Delta t]:~ u(t) = c_u + G_u \alpha_{u,j}
		\label{eq:alpha}
	\end{equation}
	with $j = 1,\dots,t_{\text{\normalfont end}}/\Delta t$, the approximate solution $\mu(t,x_0,u(\cdot))$ and the time-varying error bound $\mathcal{E}(t)$ required for Alg.~\ref{alg:reach} to be dependency preserving according to Def.~\ref{def:dependency} are given as
	\begin{equation}
	\begin{split}
		\mu(t,x_0,u(\cdot)) &= \begin{cases} [G_{i1(\cdot,\mathbf{H}_{i})}~\mathbf{0}] \, \widetilde{\alpha}, & t~\text{\normalfont mod}~\Delta t = 0 \\ [G_{i2(\cdot,\mathbf{N}_{i})}~\mathbf{0}] \, \widetilde{\alpha}, & \text{\normalfont otherwise} \end{cases} \\
		\mathcal{E}(t) &= \begin{cases} \langle c_{i1},G_{i1(\cdot,\mathbf{K}_{i})}\rangle_Z, & t~\text{\normalfont mod}~\Delta t = 0 \\ \langle c_{i2},G_{i2(\cdot,\mathbf{M}_{i})}\rangle_Z, & \text{\normalfont otherwise}\end{cases}
	\end{split}
	\label{eq:dependency}
	\end{equation}
	with 
	\begin{equation*}
		i = \lfloor t/\Delta t  \rfloor, ~~~~\mathcal{R}(t_i) = \langle c_{i1},G_{i1} \rangle_Z, ~~~~\mathcal{R}(\tau_i) = \langle c_{i2},G_{i2} \rangle_Z,
	\end{equation*}
	where the vector 
	\begin{equation}
		\widetilde{\alpha} = [\alpha_x \concat \alpha_{u,1} \concat \dots \concat \alpha_{u,{t_{\text{\normalfont end}}/\Delta t}}]
		\label{eq:alphaVec}
	\end{equation} 
	stores the factor values that define the initial state and input signal in \eqref{eq:alpha} and the tuples
	\begin{equation}
	\begin{split}
		& \mathbf{H}_{i} = (1,\dots,\gamma_x,\gamma_x + 1,\dots,\gamma_x + \gamma_u \, i ), ~ \mathbf{K}_{i} = (1,\dots,\gamma_{i1}) \setminus \mathbf{H}_{i} \\
		\mathbf{N}_{i} = &( 1, \dots,\gamma_x,2 \, \gamma_x + 2,\dots,2 \, \gamma_x + 2 + \gamma_u \, (i+1) ), ~ \mathbf{M}_{i} = (1,\dots,\gamma_{i2}) \setminus \mathbf{N}_{i}  \\
	\end{split}
	\label{eq:indices}
	\end{equation}
	store the indices of specific zonotope generators. 
	\label{prop:dependency}
\end{proposition}
\begin{proof}
	Instead of directly mapping initial states and inputs to reachable states, our approximate solution $\mu(t,x_0,u(\cdot))$ relies on an intermediate representation by zonotope factors $\widetilde{\alpha}$. We therefore first represent the initial state $x_0$ as well as the input signal $u(\cdot)$ by the corresponding zonotope factors $\widetilde{\alpha}$ using the zonotopes $\mathcal{X}_0$ and $\mathcal{U}$, and then map these zonotope factors to reachable states using the zonotopes that represent the reachable set $\mathcal{R}(t)$. In particular, according to Def.~\ref{def:zonotope}, every point $x \in \mathcal{Z}$ inside a zonotope $\mathcal{Z} = \langle c,G \rangle_Z \subset \Rn$ can equivalently be represented by the corresponding zonotope factors $\alpha \in \R^\gamma$ using the relation $x = c + G\, \alpha$. This allows us to equivalently represent the initial state $x_0 \in \mathcal{X}_0$ and the piecewise constant input signal $u(\cdot) \in \mathcal{U}$ by the vector of zonotope factors $\widetilde{\alpha}$ in \eqref{eq:alphaVec} according to \eqref{eq:alpha}. 
Since Alg.~\ref{alg:reach} is composed of the set operations linear map, Minkowski sum, and convex hull, which are all dependency preserving for zonotopes according to \cite[Tab.~1]{Kochdumper2020c}, it holds that Alg.~\ref{alg:reach} is dependency preserving, too \cite[Lemma~2]{Kochdumper2020c}. This means that we can directly insert $\widetilde{\alpha}$ into the zonotopes that represent the reachable set $\mathcal{R}(t)$ in \eqref{eq:dependency} to map from factors to reachable states, where we use the time point reachable set $\mathcal{R}(t_i)$ if the time $t$ is equivalent to $t_i$ and the time interval reachable set $\mathcal{R}(\tau_i)$ if the time $t$ is between two time points. When mapping from factors to reachable states we have to distinguish between generators that correspond to initial states and inputs, and generators that represent uncertainty arising from the difference between constant and time-varying inputs $\mathcal{D}$, the curvature enclosure $\mathcal{C}_i$, or the over-approximation from the enclosure of the convex hull of two zonotopes computed using \eqref{eq:zonoConvHull}. We achieve this with the tuples $\mathbf{H}_i,\mathbf{N}_i,\mathbf{K}_i,\mathbf{M}_i$ in \eqref{eq:indices}, where $\mathbf{H}_i,\mathbf{N}_i$ store the indices of the generators that correspond to initial states as well as inputs and the tuples $\mathbf{K}_i,\mathbf{M}_i$ store the indices of the generators that represent uncertainty.
\end{proof}

In addition to dependency preservation, another important property of Alg.~\ref{alg:reach} is that the computed enclosure of the reachable set converges to the exact reachable set for $\Delta t \to 0$. The corresponding proof is provided in \ref{sec:appendixB}. As a direct consequence, also the size of the error $\mathcal{E}(t)$ in Prop.~\ref{prop:dependency} converges to 0 for $\Delta t \to 0$, so that the approximate solution $\mu(t,x_0,u(\cdot))$ converges to the exact solution $\xi(t,x_0,u(\cdot))$. While a large over-approximation error for reachability analysis might prevent the successful verification of the temporal logic specification, improving the accuracy to achieve a smaller error increases the computation time. To resolve this trade-off, our strategy for automated verification is to start with a quite inaccurate reachable set enclosure that can be computed very efficiently, and then iteratively refine the accuracy until the temporal logic specification can be either verified or falsified.

\section{Model Checking}
\label{sec:modelChecking}

The main disadvantage of model checking using RTL is that the intersection checks in \eqref{eq:RTLfinal} are evaluated in a pure true/false manner, losing all information about which initial states or inputs result in an intersection. As visualized in Fig.~\ref{fig:example}, the model checking procedure therefore also considers so-called spurious traces, which are contained in the reachable set but are not consistent with the system dynamics. In this work we solve this problem by explicitly keeping track which initial states and which inputs result in an intersection.

In particular, we approximate the set of all possible time-varying inputs $u(\cdot) \in \mathcal{U}$ with the set of all possible piecewise constant inputs, because our reachability algorithm Alg.~\ref{alg:reach} is dependency preserving for piecewise constant inputs. This approximation does not negatively affect the completeness of our model checking algorithm, since in the limit case $\Delta t \to 0$ piecewise constant inputs approximate time-varying inputs arbitrary well. For our model checking algorithm, we have to consider the set containing all possible combinations of initial states and piecewise constant inputs, which is given as $\mathcal{X}_0 \times \mathcal{U} \times \dots \times \mathcal{U}$. As shown in \eqref{eq:alpha}, since both the initial set and the input set are represented by zonotopes, we can equivalently represent every initial state and every piecewise constant input signal by the corresponding zonotope factors $\alpha_x$ and $\alpha_{u,i}$. Therefore, the set $\mathcal{X}_0 \times \mathcal{U} \times \dots \times \mathcal{U}$ containing all possible combinations of initial states and inputs can equivalently be represented in the space of zonotope factors by the hypercube $\widetilde \alpha = [-\mathbf{1},\mathbf{1}]$ with $\widetilde \alpha$ defined as in \eqref{eq:alphaVec}. Since the hypercube $\widetilde \alpha = [-\mathbf{1},\mathbf{1}]$ is computationally much easier to handle than the zonotope $\mathcal{X}_0 \times \mathcal{U} \times \dots \times \mathcal{U}$, our model checking algorithm represents initial states and inputs in factor space $\widetilde \alpha = [-\mathbf{1},\mathbf{1}]$.

 
As visualized in Fig.~\ref{fig:temporalLogic}, for each intersection between a reachable set $\widetilde{\mathcal{R}}(j \, \Delta t/2)$ and a polytope $\mathcal{P}_{hjkl}$ we can divide the domain $\widetilde{\alpha} \in [-\mathbf{1},\mathbf{1}]$ into factors that potentially lead to an intersection and factors that are guaranteed to not result in an intersection with the polytope $\mathcal{P}_{hjkl}$. The corresponding polytopes $\fac$ representing factors that potentially result in a violation can be computed as follows:

\begin{figure}[!tb] 
	\centering
	\includegraphics[width = 0.9\textwidth]{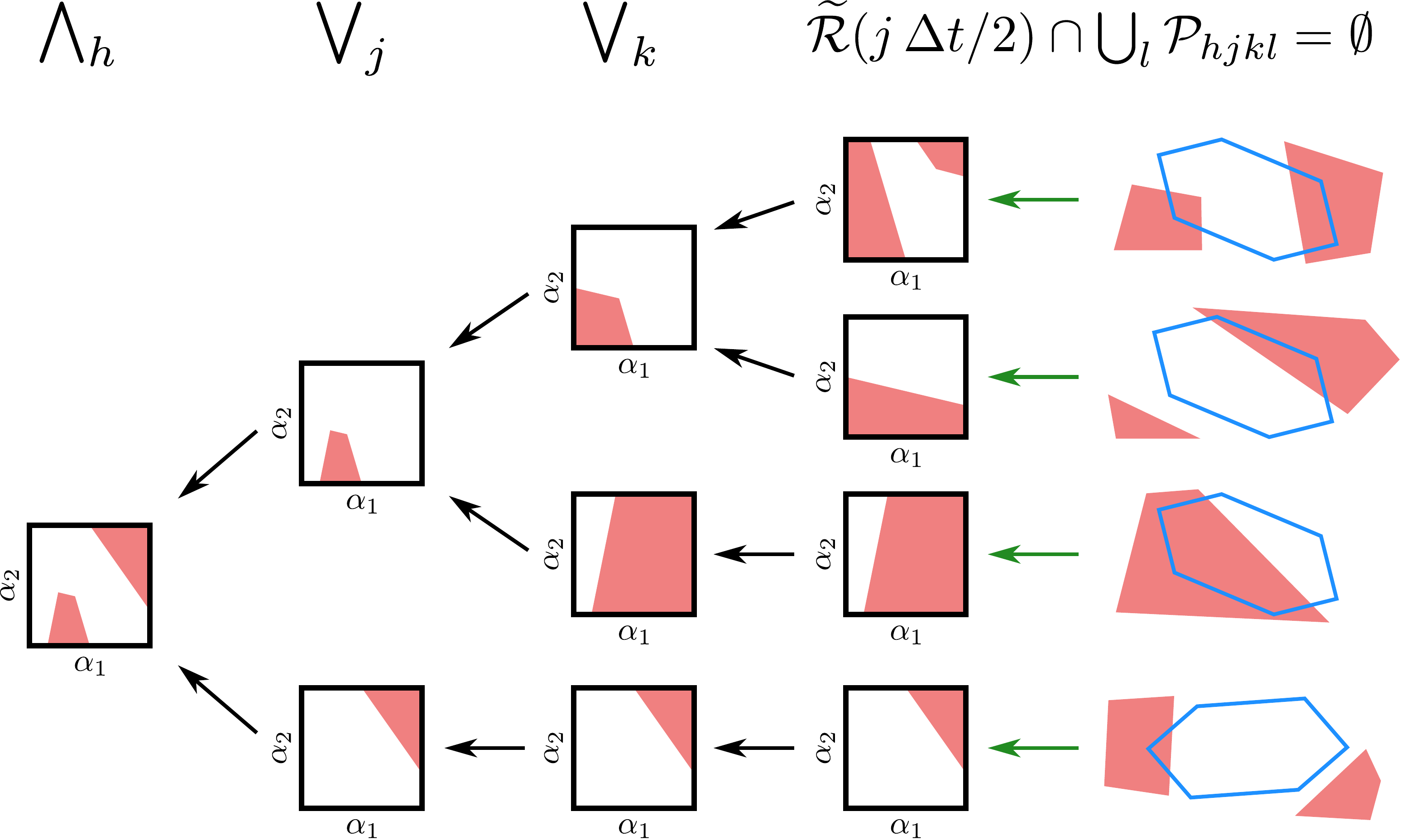}
	\caption{Schematic visualization of the model checking approach presented in Sec.~\ref{sec:modelChecking}, where the reachable sets $\widetilde{\mathcal{R}}(j \, \Delta t/2)$ are visualized in blue and the polytopes $\mathcal{P}_{hjkl}$ as well as the corresponding sets of unsafe factors $\fac$ are depicted in red. Moreover, the green arrows represent the mapping from the state space to the space of zonotope factors $\widetilde{\alpha} \in [-\mathbf{1},\mathbf{1}]$.}
	\label{fig:temporalLogic}
\end{figure}

\begin{algorithm}[!tb]
	\caption{Model Checking} \label{alg:modelCheck}
	\vspace{2pt}
	{\raggedright \textbf{Input:} Linear system $\dot x = A x + B u$, STL specification $\varphi$, initial set $\initSet$, input set $\mathcal{U}$, final time $t_{\text{end}}$, time step size $\Delta t$, truncation order $\kappa$.
	
	\textbf{Output:} List $\mathbf{L}$ containing polytopes that represent zonotope factors of the initial and input set which violate the specification.
	}
	\begin{algorithmic}[1]
		\State $\mathcal{R}(t_0),\mathcal{R}(\tau_0),\dots,\mathcal{R}(t_{\text{end}}) \gets $ comp. reach. set with $\Delta t$ and $\kappa$ using Alg.~\ref{alg:reach}	
		\State $\bigwedge_h \bigvee_j \bigvee_k (\widetilde{\mathcal{R}}(j \,\Delta t/2) \cap \bigcup_l \mathcal{P}_{hjkl} = \emptyset) \gets$ convert $\varphi$ to RTL using $\Delta t$ 
		\Statex[12] ~(see \eqref{eq:RTLfinal} and \cite{Roehm2016b})
		\State $\mathbf{L} \gets \emptyset$
		\For{$h \gets 0$ to $H$}
			\State $\mathbf{V} \gets [-\mathbf{1},\mathbf{1}]$
			\For{$j \gets 0$ to $J$ and $k \gets 0$ to $K$}
					\State $\mathbf{W} \gets \emptyset$
					\For{$l \gets 0$ to $L$}
						\If{$\widetilde{\mathcal{R}}(j \,\Delta t/2) \cap \mathcal{P}_{hjkl} \neq \emptyset$}
							\State $\fac \gets $ unsafe factors from $ \widetilde{\mathcal{R}}(j \,\Delta t/2) \cap \mathcal{P}_{hjkl}$ using Prop.~\ref{prop:constraints}
							\State $\mathbf{W} \gets (\mathbf{W},\fac)$
						\EndIf
					\EndFor
					\State $\mathbf{V} \gets (\mathbf{V}_{(1)} \cap \mathbf{W}_{(1)}, \dots, \mathbf{V}_{(|\mathbf{V}|)} \cap \mathbf{W}_{(|\mathbf{W}|)})$
					\State $\mathbf{V} \gets $ remove empty polytopes from $\mathbf{V}$ \label{line:empty}
			\EndFor
			\State $\mathbf{L} \gets (\mathbf{L}, \mathbf{V})$
		\EndFor	
	\end{algorithmic}
\end{algorithm}

\begin{proposition}
	Given the time point reachable set $\mathcal{R}(t_i) = \langle c_{i1},G_{i1}\rangle_Z$ and time interval reachable set $\mathcal{R}(\tau_i) = \langle c_{i2},G_{i2} \rangle_Z \subset \mathbb{R}^n$ computed using Alg.~\ref{alg:reach} as well as the polytope $\mathcal{P} = \langle C,d \rangle_P \subseteq \mathbb{R}^n$, it holds that for all zonotope factors $\widetilde \alpha$ that are not located in the polytopes 
	\begin{align*}
		\fac_{i1} & = \bigg \langle \big[C G_{i1(\cdot,\mathbf{H}_i)}~\mathbf{0} \big],~d - C c_{i1} + \sum_{j \in \mathbf{K}_i} | C G_{i1(\cdot,j)} | \bigg \rangle_P \\
		\fac_{i2} & = \bigg \langle \big[C G_{i2(\cdot,\mathbf{N}_i)}~\mathbf{0} \big],~d - C c_{i2} + \sum_{j \in \mathbf{M}_i} | C G_{i2(\cdot,j)} | \bigg \rangle_P
	\end{align*}
	there is no intersection between the corresponding reachable states and $\mathcal{P}$:
	\begin{equation*}
	\begin{split}
		\forall \widetilde{\alpha} \in [-\mathbf{1},\mathbf{1}]: ~~ &\big( \widetilde{\alpha} \not \in \fac_{i1} \big) \Rightarrow \big( \xi(t_i,x_0,u(\cdot)) \not \in \mathcal{P} \big) \\
		\forall \widetilde{\alpha} \in [-\mathbf{1},\mathbf{1}],\, \forall t \in \tau_i: ~~ &\big( \widetilde{\alpha} \not \in \fac_{i2} \big) \Rightarrow \big( \xi(t,x_0,u(\cdot)) \not \in \mathcal{P} \big),
	\end{split} 
	\end{equation*}
	where $\widetilde{\alpha} = [\alpha_x \concat \alpha_{u,1} \concat \dots \concat \alpha_{u,t_{\text{\normalfont end}}/\Delta t}]$, $x_0$ and $u(\cdot)$ are defined as in \eqref{eq:alpha}, and the tuples $\mathbf{H}_i$,$\mathbf{K}_i$,$\mathbf{N}_i$ and $\mathbf{M}_i$ are defined as in \eqref{eq:indices}.
	\label{prop:constraints}
\end{proposition}
\begin{proof}
	According to Prop.~\ref{prop:dependency} and Def.~\ref{def:dependency} it holds that 
	\begin{equation}
		\xi(t_i,x_0,u(\cdot)) \in \underbrace{[G_{i1(\cdot,\mathbf{H}_{i})}~\mathbf{0}] \, \widetilde{\alpha}}_{\mu(t,x_0,u(\cdot))}  \oplus \underbrace{\langle c_{i1},G_{i1(\cdot,\mathbf{K}_{i})}\rangle_Z}_{\mathcal{E}(t)}.
		\label{eq:proof1}
	\end{equation}
	Moreover, a reachable state $\xi(t_i,x_0,u(\cdot))$ is according to Def.~\ref{def:polytope} contained in the polytope $\mathcal{P} = \langle C,d\rangle_P$ if $C \, \xi(t_i,x_0,u(\cdot)) \leq d$ holds. Combining this with \eqref{eq:proof1} yields the condition
	\begin{equation*}
		\exists x \in \langle C c_{i1}, C G_{i1(\cdot,\mathbf{K}_i)} \rangle_Z:~ [C G_{i1(\cdot,\mathbf{H}_i)} ~ \mathbf{0}] \, \widetilde \alpha + x \leq d.
	\end{equation*}
	Finally, we bring the zonotope $\langle C c_{i1}, C G_{i1(\cdot,\mathbf{K}_i)} \rangle_Z$ that represents uncertainty in the computed enclosure of the reachable set to the other side of the inequality using the formula for the interval enclosure of a zonotope \cite[Prop. 2.2]{Althoff2010a}, which yields the polytope $\fac_{i1}$. In a similar way we obtain the result for the time interval reachable set.
\end{proof}
As shown in Fig.~\ref{fig:temporalLogic}, the conjunctions in \eqref{eq:RTLfinal} correspond to intersections of the sets containing safe factors, which is equivalent to a union of the polytopes $\fac$ representing potentially unsafe factors. Similarly, the disjunctions in \eqref{eq:RTLfinal} correspond to a union of safe factors which is equivalent to an intersection of unsafe factors.
Alg.~\ref{alg:modelCheck} summarizes the corresponding model checking procedure. The system satisfies the STL specification if the conjunctions and disjunctions cancel out all unsafe sets $\fac$, which corresponds to an empty list $\mathbf{L} = \emptyset$. 

\begin{algorithm}[!tb]
	\caption{Automated Verification} \label{alg:verify}
	\vspace{2pt}
	{\raggedright \textbf{Input:} Linear system $\dot x = A x + B u$, STL specification $\varphi$, initial set $\initSet$, input set $\mathcal{U}$.
	
	\textbf{Output:} Safe ($\mathcal{R}^{\text{\normalfont e}}(t) \vDash \varphi$) or unsafe ($\mathcal{R}^{\text{\normalfont e}}(t) \nvDash \varphi$).
	
	}
	\begin{algorithmic}[1]
		\State $t_{\text{end}} \gets $ final time for temporal logic formula $\varphi$	\label{line:time}
		\State $\Delta t \gets t_{\text{end}}$ \label{line:init}
		\MRepeat
			\State $\kappa \gets$ increase $\kappa$ until $1 - ||\mathbfcal{T}^{(\kappa)}||_F / ||\mathbfcal{T}^{(\kappa+1)}||_F \leq 10^{-10}$ (see \eqref{eq:curv})
			\State $\mathbf{L} \gets$ model check $\varphi$ with parameters $\Delta t$ and $\kappa$ using Alg.~\ref{alg:modelCheck} \label{line:safeStart}
			\If{$\mathbf{L} = \emptyset$} \label{line:safe}
				\State \Return safe \label{line:safeEnd}
			\EndIf
			\State $\textbf{L} \gets$ model check $\neg \varphi$ with parameters $\Delta t$ and $\kappa$ using Alg.~\ref{alg:modelCheck} \label{line:unsafe1}
			\If{$\exists \widetilde \alpha \in [-\mathbf{1},\mathbf{1}]: ~ \widetilde \alpha \not \in \fac_j ~ \forall \fac_j \in \mathbf{L}$} \label{line:unsafe}
				\State \Return unsafe \label{line:unsafe2}
			\EndIf
			\State $\Delta t \gets \Delta t /2$ \label{line:refine}
		\EndRepeat		
	\end{algorithmic}
\end{algorithm}

\section{Automated Verification}
\label{sec:verify}

The overall verifier is summarized in Alg.~\ref{alg:verify}. Since all temporal operators are bounded in time, the STL formula $\varphi$ is restricted to a finite time interval, which we determine in Line~\ref{line:time}. We then first initialize the time step size $\Delta t$ with the overall time horizon $t_{\text{end}}$ (see Line \ref{line:init}), and afterward refine it in each iteration of the main loop (see Line~\ref{line:refine}). For tuning the truncation order $\kappa$ we apply the strategy from \cite[Sec.~IV.B]{Wetzlinger2022}, which is based on the interval matrices $\mathbfcal{T}^{(o)}$ in \eqref{eq:curv}.
Starting from $\kappa = 2$, we increase $\kappa$ until the relative change between the Frobenius norms $|| \mathbfcal{T}^{(\kappa)} ||_F$ and $|| \mathbfcal{T}^{(\kappa+1)} ||_F$ computed according to \cite[Thm.~10]{Farhadsefat2011} is smaller than $10^{-10}$.
As explained in Sec.~\ref{sec:modelChecking}, the system satisfies the STL formula if the list $\mathbf{L}$ containing the sets of unsafe factors is empty, which we check in Line~\ref{line:safe}. On the other hand, the system violates the STL formula $\varphi$ if there exists a single initial state and input signal that satisfy the negated formula $\neg \varphi$. Consequently, to prove $\mathcal{R}(t) \not \vDash \varphi$, we first run Alg.~\ref{alg:modelCheck} to check if $\mathcal{R}(t)$ satisfies $\neg \varphi$ in Line~\ref{line:unsafe1}, which yields a list of safe sets $\mathbf{L}$. Then, we need to show that there exists a vector of factors $\widetilde \alpha \in [-\mathbf{1},\mathbf{1}]$ that does not intersect any of the safe sets $\fac_j = \langle C_j,d_j\rangle_P \in \mathbf{L}$ (see Line~\ref{line:unsafe}), and therefore corresponds to an unsafe initial state and input signal. This can be realized by solving the following mixed-integer linear program:
\begin{equation}
\begin{split}
		\min_{\widetilde \alpha \in [-\mathbf{1},\mathbf{1}]} \| \widetilde \alpha \|_1  ~~ \text{subject to} ~~ & \forall j \in \{1,\dots,|\mathbf{L}|\}, ~ \forall k \in \{ 1,\dots,\poly_j\}: \\
		& C_{j(\cdot,k)} \, \widehat{\alpha}_{jk} > \lambda_{jk} \, d_{j(k)}, ~ - \mathbf{1} \, \lambda_{jk} \leq \widehat{\alpha}_{jk} \leq \mathbf{1} \, \lambda_{jk}, \\
		& \textstyle \lambda_{jk} \in \{0,1\}, ~ \widetilde \alpha = \sum_{k=1}^{s_j} \widehat{\alpha}_{jk}, ~\sum_{k=1}^{s_j} \lambda_{jk} = 1.
\end{split}
\label{eq:mixedInteger}
\end{equation}
If \eqref{eq:mixedInteger} has a feasible solution, the system violates the specification $\varphi$. In addition, the optimal solution $\widetilde \alpha = [\alpha_x \concat \alpha_{u,1} \concat \dots \concat \alpha_{u,t_{\text{end}}/\Delta t}]$ for \eqref{eq:mixedInteger} defines a falsifying trajectory
\begin{equation*}
	x(t_{i+1}) = e^{A \Delta t} x(t_i) + T(c_u + G_u \alpha_{u,i}), \quad \quad i = 0,\dots, t_{\text{end}}/\Delta t
\end{equation*}
with initial state $x(t_0) = c_x + G_x \alpha_x$.

The verification problem for linear systems against STL specifications is undecidable in general since it is not possible to compute the exact reachable set \cite{Lafferriere2001}. However, undecidable verification problems only occur if the reachable set is located exactly on the decision boundary between safe and unsafe, and are therefore quite rare in practice. Consequently, most verification problem are decidable, and Alg.~\ref{alg:verify} returns the correct result in finite time in these cases:

\begin{theorem}
	If the verification problem defined by the linear system $\dot x = Ax + Bu$, the STL specification $\varphi$, the initial set $\mathcal{X}_0$, and the input set $\mathcal{U}$ represents a decidable problem instance, Alg.~\ref{alg:verify} will terminate in finite time.
\end{theorem}
\begin{proof}
	Our verification approach contains the following three sources of over-approximation errors:
	\begin{enumerate}
	  \item The error from the conservative enclosure of the reachable set.
	  \item The error from the time-discretization of the STL formula during the conversion to RTL.
	  \item The error through the approximation of time-varying inputs with piecewise constant inputs.
	\end{enumerate}
	As we show in \ref{sec:appendixB}, the over-approximation error from reachability analysis converges to 0 for $\Delta t \to 0$. Similarly, the error from the time-discretization during conversion to RTL becomes 0 for $\Delta t \to 0$ \cite{Roehm2016b}, and the space of piecewise constant inputs contained inside the input set $\mathcal{U}$ converges to the space of time-varying input signals contained in $\mathcal{U}$ for $\Delta t \to 0$. Consequently, all errors converge to 0 for $\Delta t \to 0$, so that Alg.~\ref{alg:verify} will always be able to find a $\Delta t$ small enough to either verify or falsify the specification.
\end{proof}
If the verification problem instance is not decidable, Alg.~\ref{alg:verify} does not terminate. To accelerate the computation, we additionally implemented two improvements for our basic verification approach presented in this section: 
\begin{itemize}
	\item[$\bullet$] We already evaluate all predicates prior to the conversion to RTL and substitute predicates that are either satisfied or violated for the whole reachable set by true/false.  
	\item[$\bullet$] We omit the conversion to conjunctive normal form performed during conversion to RTL for cases where it does not bring a benefit. In particular, this is the case if the conversion does not result in a disjunction of atomic predicates.
\end{itemize}
Finally, Alg.~\ref{alg:verify} naturally splits the domain $\widetilde \alpha \in [-\mathbf{1},\mathbf{1}]$ corresponding to the initial set and space of input signals into safe and unsafe subsets. We can therefore directly use the resulting list of unsafe sets $\mathbf{L}$ to identify initial states and input signals that are safe or to cut away all reachable states that violate the temporal logic specification.


\section{Numerical Examples}
\label{sec:numEx}

We now demonstrate the performance of our automated verifier on several challenging benchmark systems. All computations are carried out in MATLAB on a 2.9GHz quad-core i7 processor with 32GB memory. Moreover, we integrated the implementation of our automated verifier into the open-source reachability toolbox CORA \cite{Althoff2015a}, and published a repeatability package that reproduces all results presented in this section on CodeOcean\footnote{\url{https://codeocean.com/capsule/4858011/tree/v1}}.

\subsection{Comparison with Other Approaches}

We first compare our method with the tool SMT{\sc mc} \cite{Yu2022}, which implements formal verification using SMT solving \cite{Lee2021,Bae2019}. For the comparison we consider the system and STL specification in Fig.~\ref{fig:example}. Since SMT{\sc mc} does not support systems with inputs, we slightly modify the system by keeping the input $u$ constant over time, so that it can be modeled as an additional state with dynamics $\dot u = 0$. Even for this very simple system SMT{\sc mc} needs 1462 seconds to verify that the system is safe, where we used a time step size of $\Delta t = 0.1$. Our approach, on the other hand, verifies safety in only 0.25 seconds and is in addition fully automatic, so that we do not have to manually select a suitable time step size. Moreover, for the more complex systems discussed in Sec.~\ref{subsec:robot} and Sec.~\ref{subsec:arch}, SMT{\sc mc} is not able at all to provide a verification result in reasonable time, where we aborted the analysis after one hour. 

\begin{figure}[!tb]
  \centering
  \includegraphics[width = 0.98\columnwidth]{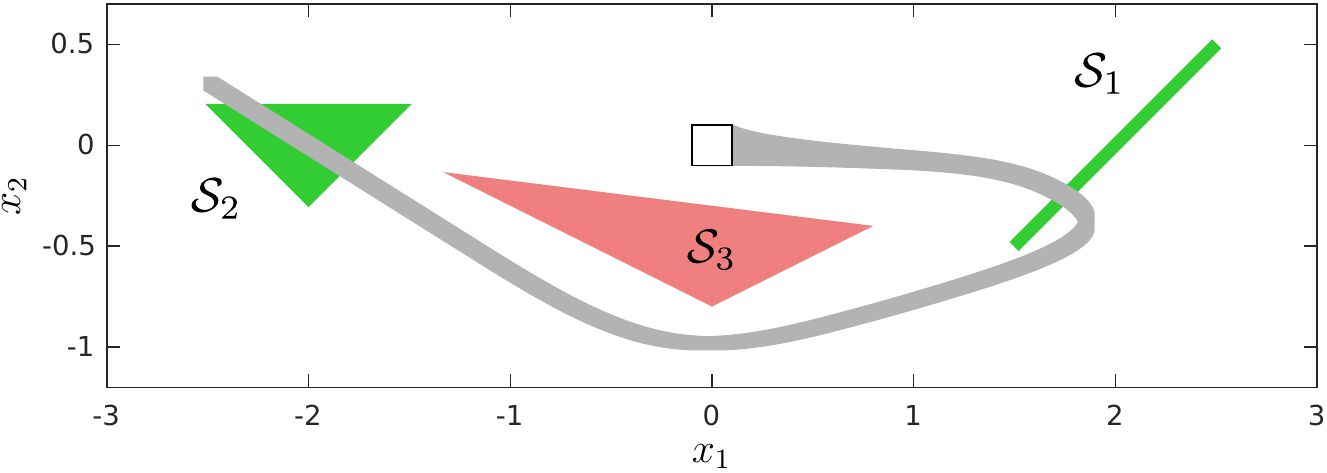}
  \caption{Reachable set for the mobile robot, where the initial set is shown in white with a black border and the sets $\mathcal{S}_1$, $\mathcal{S}_2$, $\mathcal{S}_3$ are visualized in green and red.}
  \label{fig:mobileRobot}
\end{figure}

\subsection{Mobile Robot}
\label{subsec:robot}

One typical application for temporal logic is the formulation of tasks for mobile robots. Given the sets $\mathcal{S}_1$, $\mathcal{S}_2$, $\mathcal{S}_3$ visualized in Fig.~\ref{fig:mobileRobot}, we consider the following reach-avoid task specified in natural language: 
\begin{center}
"Visit $\mathcal{S}_1$ during the first 4 seconds and afterward visit  $\mathcal{S}_2$ \\ $~~~~~~$exactly 6 seconds after $\mathcal{S}_1$ while avoiding $\mathcal{S}_3$ at all times."
\end{center}
This corresponds to the following temporal logic specification:
\begin{equation*}
	\varphi = \finally_{[0,4]} \big( x \in \mathcal{S}_1 \wedge \nextOp_{6} \, x \in \mathcal{S}_2 \big) \wedge \globally_{[0,10]} \, x \not \in \mathcal{S}_3.
\end{equation*}	
We model the dynamics of the mobile robot by a double-integrator for the x- and y-position, which yields the system matrices $A = [\mathbf{0} ~ [I_2 \concat\mathbf{0}]]$ and $B = [\mathbf{0} \concat I_2]$. Moreover, we consider that the mobile robot tracks a reference trajectory $x_{\text{ref}}(t)$ that corresponds to the piecewise constant control input $u_{\text{ref}}(t)$ using the feedback control law $u_{\text{ctr}}(t) = u_{\text{ref}}(t) + K(x(t) - x_{\text{ref}}(t))$. The corresponding feedback matrix $K \in \R^{2 \times 4}$ is determined by applying an LQR control approach with state weighting matrix $Q = I_4$ and input weighting matrix $R = 0.1 \cdot I_2$ to the open-loop system. Overall, the dynamics of the controlled system is given as
\begin{equation*}
	\begin{bmatrix} \dot x(t) \\ \dot x_{\text{ref}}(t) \end{bmatrix}
	= \begin{bmatrix} A \hspace{-2pt} + \hspace{-2pt} BK & -BK \\ \mathbf{0} & A \end{bmatrix} \begin{bmatrix} x(t) \\ x_{\text{ref}}(t) \end{bmatrix}
	+ \begin{bmatrix} B & B \\ B & \mathbf{0} \end{bmatrix} u(t).
\end{equation*}
Moreover, the initial set is $\initSet = [-0.1,0.1] \si{\meter} \times [-0.1,0.1]\si{\metre} \times \mathbf{0}$ and the set of uncertain inputs is $\mathcal{U} = u_{\text{ref}}(t) \times [-0.1,0.1]\si{\meter \per \square \second} \times [-0.1,0.1] \si{\meter \per \square \second}$.

The resulting reachable set shown in Fig.~\ref{fig:mobileRobot} demonstrates that the mobile robot satisfies the temporal logic specification $\varphi$. However, since at no point in time the reachable set is fully contained inside the set $\mathcal{S}_1$, the original RTL approach \cite{Roehm2016b} always conservatively classifies the system as unsafe, no matter how often we refine the tightness of the reachable set enclosure. On the other hand, our verifier summarized in Alg.~\ref{alg:verify} fully automatically proves that the system is safe in only $2.89$ seconds.

\begin{table*}[!tb]
\begin{center}
\caption{Computation time in seconds for our automated verifier on different benchmarks from the ARCH competition, where $n$ is the number of states and $m$ the number of inputs of the system.}
\label{tab:ARCH}
\renewcommand{\arraystretch}{1.2}
\footnotesize
\begin{tabular}{ l c c c c c}
 \toprule
 Benchmark \hspace{-10pt} & $n$ & \hspace{-5pt}$m$\hspace{-5pt} & Specification & \hspace{-5pt}Safe?\hspace{-5pt} & Time \\ \midrule
\multirow{2}{*}{BLDC01} & \multirow{2}{*}{49} & \multirow{2}{*}{0} & $\finally_{[0,0.2]}\big( \globally_{[0,0.18]}\, x_{40} < 0.0046 \vee \nextOp_{0.2} \,x_{40} > 0.006 \big)$ & \cmark & 57.7 \\ 
& & & $\finally_{[0,0.2]}\big( \globally_{[0,0.19]}\, x_{40} < 0.0046 \vee \nextOp_{0.2} \, x_{40} > 0.006 \big)$ & \xmark & 25.2 \\ \midrule
\multirow{2}{*}{BLDF01} & \multirow{2}{*}{48} & \multirow{2}{*}{1} & $\finally_{[0,0.2]}\big( \globally_{[0,0.18]}\, x_{40} < 0.0046 \vee \nextOp_{0.2} \,x_{40} > 0.006 \big)$ & \cmark & 59.8 \\ 
& & & $\finally_{[0,0.2]}\big( \globally_{[0,0.19]}\, x_{40} < 0.0046 \vee \nextOp_{0.2} \, x_{40} > 0.006 \big)$ & \xmark & 32.5 \\ \midrule
\multirow{2}{*}{CBC01} & \multirow{2}{*}{201} & \multirow{2}{*}{0} & $x_{110} < 69 ~ \until_{[0,0.004]} \, (x_{140} > 69 \vee x_{140} < -60)$ & \cmark & 5.2 \\ 
& & & $x_{110} < 68 ~ \until_{[0,0.004]} \, (x_{140} > 69 \vee x_{140} < -60)$ & \xmark & 9.1 \\ \midrule
\multirow{2}{*}{CBF01} & \multirow{2}{*}{200} & \multirow{2}{*}{1} & $x_{110} < 69 ~ \until_{[0,0.004]} \, (x_{140} > 69 \vee x_{140} < -60)$ & \cmark & 17.9 \\ 
& & & $x_{110} < 68 ~ \until_{[0,0.004]} \, (x_{140} > 69 \vee x_{140} < -60)$ & \xmark & 30.8 \\ \midrule
\multirow{2}{*}{HEAT01} & \multirow{2}{*}{125} & \multirow{2}{*}{0} & $\globally_{[4,8]} \, x_{10} > 0.1  \vee \finally_{[8,14]} \, x_{10} < 0.1$ & \cmark & 9.1 \\ 
& & & $\globally_{[4,8]} \, x_{10} > 0.1  \vee \finally_{[8,13]} \, x_{10} < 0.1$ & \xmark &  20.1 \\ \midrule
\multirow{2}{*}{HEAT02} & \multirow{2}{*}{1000} & \multirow{2}{*}{0} & $\globally_{[3,4]} \, x_{151} > 0.07  \vee \finally_{[4,6]} \, x_{151} < 0.07$ & \cmark & 52.8 \\ 
& & & $\globally_{[3,4]} \, x_{151} > 0.07  \vee \finally_{[4,5]} \, x_{151} < 0.07$ & \xmark & 50.8 \\ 
 \bottomrule
\end{tabular}
\end{center}
\end{table*}

\subsection{ARCH Benchmarks}
\label{subsec:arch}

To demonstrate that our approach also scales to high-dimensional systems we consider the benchmarks from the 2022 ARCH competition \cite{ARCH22linear}, which represent the limit of complexity that can be handled by state-of-the-art reachability tools. 
For each benchmark, we introduce one safe and one unsafe temporal logic specification. To make the problems more challenging all specifications are close to the decision boundary, meaning that already slight modifications in the specification change the verification result from safe to unsafe or the other way round. 

The original RTL approach \cite{Roehm2016b} conservatively classifies all problem instances as unsafe, even those who satisfy the specification. As shown in Tab.~\ref{tab:ARCH}, our automated verifier on the other hand correctly verifies and falsifies all benchmarks in under one minute, even though we consider complex temporal logic specifications with nested temporal operators and high-dimensional systems with up to 1000 states. Moreover, if the system is unsafe our verifier returns a falsifying trajectory that demonstrates the safety violation, as exemplary shown for the HEAT02 benchmark on the right side of Fig.~\ref{fig:arch}. The left side of Fig.~\ref{fig:arch} highlights the difference between our method and the original RTL approach \cite{Roehm2016b}: The overall reachable set does neither satisfy the specification $\globally_{[3,4]}\, x_{151} > 0.07$ nor $\finally_{[4,6]}\, x_{151} < 0.07$, so that RTL conservatively classifies the system as unsafe. However, since all initial states that violate the specification $\globally_{[3,4]}\, x_{151} > 0.07$ satisfy the specification $\finally_{[4,6]}\, x_{151} < 0.07$, our method is able to prove that the system is safe. Finally, the decomposition of the computation time into the time spend on the different parts of Alg.~\ref{alg:verify} shown in Fig.~\ref{fig:archTimes} demonstrates that the percentage of time spend on each part heavily depends on the benchmark, and that no part clearly dominates the computation time for all benchmarks.

\begin{figure}[!tb]
  \centering
  \includegraphics[width=0.98\columnwidth]{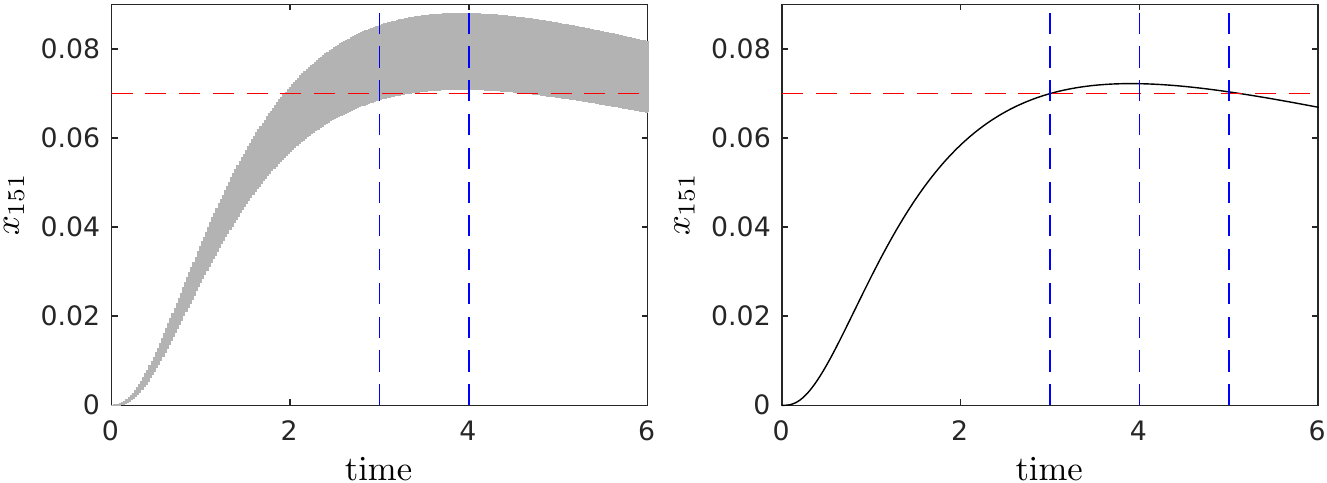}
  \caption{Reachable set (left) and falsifying trajectory (right) for the safe and the unsafe specification of the HEAT02 benchmark, where the dashed red line visualizes the boundary of the safe region and the dashed blue lines mark the points in time relevant for the temporal logic specification.}
  \label{fig:arch}
\end{figure}

\begin{figure}[!tb]
  \centering
  \includegraphics[width=0.98\columnwidth]{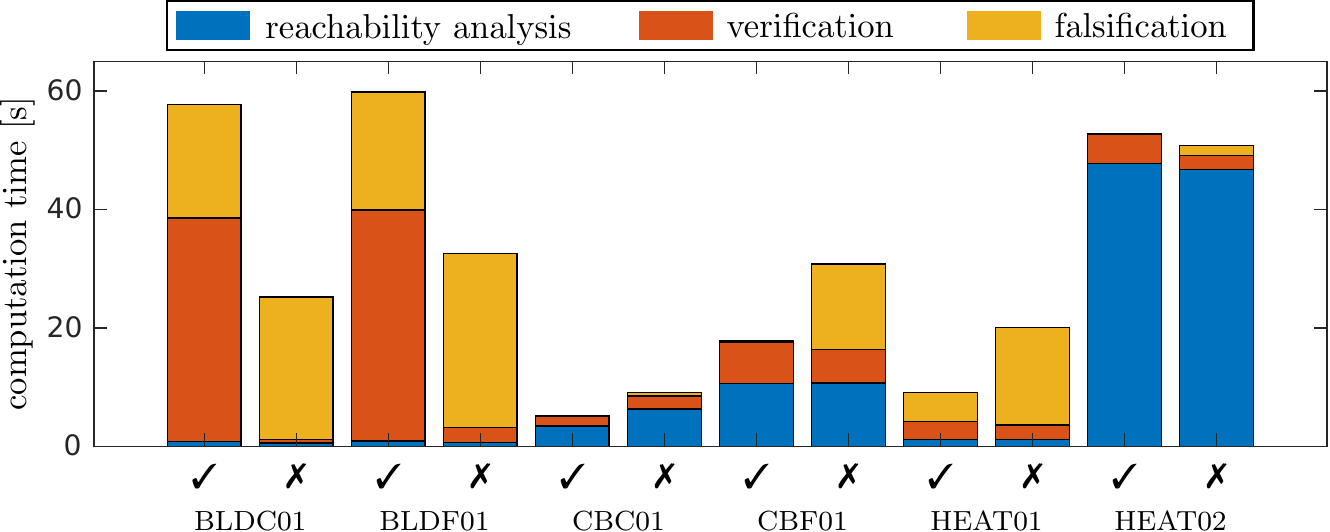}
  \caption{Computation time for our automated verifier on the ARCH benchmarks from Tab.~\ref{tab:ARCH} partitioned into the time required for reachablity analysis (see Alg.~\ref{alg:reach}), for verification (see Line~\ref{line:safeStart}-\ref{line:safeEnd} of Alg.~\ref{alg:verify}), and for falsification (see Line~\ref{line:unsafe1}-\ref{line:unsafe2} of Alg.~\ref{alg:verify}). }
  \label{fig:archTimes}
\end{figure}

\subsection{Prediction of Traffic Participants}

In addition to verifying dynamic systems against temporal logic specifications, our approach can also be used to cut away the parts of the reachable set which violate a temporal logic formula.
To demonstrate this, we consider set-based predication of traffic participants \cite{Koschi2020} in this section instead of a verification task. For set-based prediction of traffic participants, one first computes the space potentially occupied by other cars using reachability analysis, and then cuts away the regions that violate traffic rules. The trajectory of an autonomous car can then be planned to avoid the resulting regions, which guarantees safety at all times. In particular, we examine the traffic scenario shown in Fig.~\ref{fig:commonRoad}, which contains a no passing traffic sign. For this scenario the traffic rules defined by the no passing sign as well as the rule that the car is not allowed to leave the road can be formalized by the following temporal logic specification:  
\begin{equation*}
	\varphi = \globally_{[0,1]}\, (x_1 < 22 \vee x_2 < 2) \wedge \globally_{[0,1]} \, x_2 > -2 \wedge \globally_{[0,1]} \, x_2 < 6.
\end{equation*}
We model the dynamics of the car by a double-integrator for the x- and y-position, which yields the system $\dot x = A \, x + B \, u$ with $A = [\mathbf{0} ~ [I_2 \concat \mathbf{0}]]$ and $B = [\mathbf{0} \concat I_2]$. The uncertainty in the initial set $\initSet = [-0.1,0.1]\si{\meter} \times [-0.1,0.1]\si{\meter} \times [29.9,30.1]\si{\meter \per \second} \times [-0.1,0.1] \si{\meter \per \second}$ captures measurement errors in the position and velocity of the car. Moreover, the uncertain input $u(\cdot) \in \mathcal{U} = [-9,9] \si{\meter \per \square \second} \times [-9,9] \si{\meter \per \square \second}$ represents the unknown behavior of the driver, which is bounded by the physical limits of the car. For set-based prediction we execute Alg.~\ref{alg:modelCheck} with time step size $\Delta t = 0.05\si{\second}$ and truncation order $\kappa = 10$ for the negated specification $\neg \varphi$. The list $\mathbf{L}$ returned by the algorithm then contains all factor values that potentially satisfy the specification $\varphi$. A guaranteed enclosure of the set of all legal behaviors of the car is therefore given by the union of all polytopes in the list $\mathbf{L}$. Please note that the set resulting from the combination of the polytopes for the safe factors and the zonotopes for the reachable sets can be represented as a constrained zonotope \cite{Scott2016}. 

\begin{figure}[!tb]
  \centering
  \includegraphics[width=0.98\columnwidth]{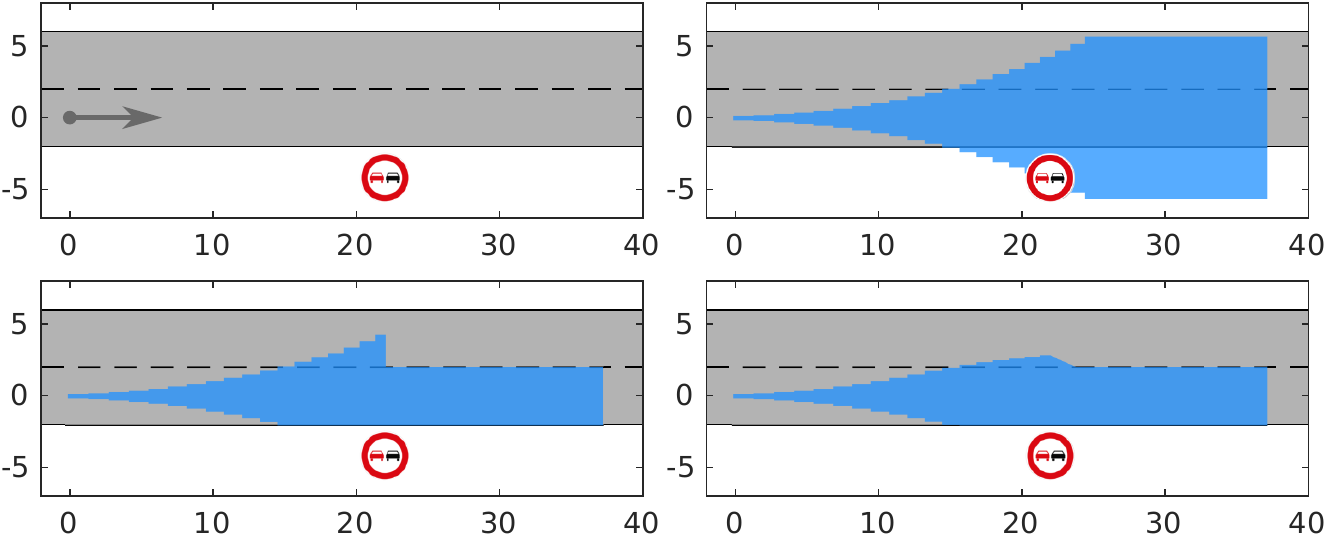}
  \caption{Visualization of set-based prediction of traffic participants showing the considered traffic scenario (top left), the overall reachable set (top right), the naive approach that simply cuts away unsafe regions (bottom left), and the set of all legal behaviors obtained with our approach (bottom right).}
  \label{fig:commonRoad}
\end{figure}

The resulting sets are visualized in Fig.~\ref{fig:commonRoad}. A naive approach is to simply cut away the unsafe regions $x_1 > 22 \wedge x_2 > 2$, $x_2 < -2$, and $x_2 > 6$. However, this ignores the dynamic behavior of the car and therefore yields a quite conservative result (Fig.~\ref{fig:commonRoad}, bottom left). Our approach, on the other hand, explicitly takes the dynamic behavior into account and therefore also cuts away regions that are guaranteed to reach an unsafe region in the future (Fig.~\ref{fig:commonRoad}, bottom right). This yields a smaller region still containing all legal behaviors, which increases the chances for finding a safe trajectory for the autonomous car.

\section{Future Work}
\label{sec:future}

While we in this work focused on linear systems for simplicity, the general framework we presented can equivalently be applied for nonlinear and hybrid systems. Therefore, the extension to nonlinear and hybrid systems is our main focus for future work, and we outline some challenges as well as potential solutions in this section.   

One key requirement for our automated verification framework is a reachability algorithm that preserves dependencies and converges to the exact reachable set if the algorithm parameters are tuned adequately. While dependency preserving reachability algorithms are readily available for nonlinear systems \cite{Kochdumper2020c}, it is more complicated to preserve dependencies for hybrid systems due to the intersections with the guard sets at discrete transitions. One promising strategy to handle those discrete transitions in a dependency-preserving fashion is the guard-mapping approach \cite{Althoff2012a}, since it circumvents the explicit computation of geometric intersections. Moreover, automated parameter tuning is significantly harder for nonlinear and hybrid systems compared to linear systems due to the increased number of algorithm parameters. A promising approach for automated parameter tuning for those systems has been published recently \cite{Wetzlinger2021}, but does not yet guarantee convergence to the exact reachable set. 

Another challenge arises from the fact that reachable sets for nonlinear and hybrid systems are in general non-convex, so that non-convex set representations such as polynomial zonotopes \cite{Kochdumper2019} or Taylor models \cite{Chen2012} are required to compute tight enclosures. In our framework, replacing the convex zonotopes we used for linear systems with those non-convex set representations yields polynomial level sets instead of polytopes for the regions of unsafe zonotope factors. With polynomial level sets the checks for empty sets in Line~\ref{line:empty} of Alg.~\ref{alg:modelCheck} are computationally more expensive, and the optimization problem in \eqref{eq:mixedInteger} becomes a mixed-integer polynomial program, which is harder to solve. Potential solutions to address these issues are using contractor programming \cite[Chapter~4]{Jaulin2006} to check if a polynomial level set is empty, and applying continuous relaxation  \cite{Hamzeei2014}
to warm-start the mixed-integer polynomial program.

\section{Conclusion}

We presented the first fully-automated verifier for linear systems that considers the very general case of specifications defined by signal temporal logic formulas. Our algorithm avoids the conservatism of the previous reachset temporal logic approach by explicitly keeping track of which initial states and which uncertain inputs satisfy or violate single parts of the formula, and is therefore guaranteed to always find the correct solution in finite time for decidable verification problem instances. As we demonstrated with numerical experiments, our automated verifier can even solve high-dimensional benchmarks with up to 1000 states as well as complex temporal logic specifications with nested temporal operators very efficiently. Other advantageous features of our algorithm are that it returns a falsifying trajectory in case of a safety violation and that it naturally divides the initial as well as the input set into parts which satisfy or violate the specification. This is for instance beneficial for set-based prediction of traffic participants, as we demonstrated with an exemplary traffic scenario. 

\vspace{0.5cm}

{\small \noindent \textbf{Acknowledgements.} This material is based upon work supported by the Air Force Office of Scientific Research and the Office of Naval Research under award numbers FA9550-19-1-0288, FA9550-21-1-0121, FA9550-23-1-0066 and N00014-22-1-2156. Any opinions, findings, and conclusions or recommendations expressed in this material are those of the authors and do not necessarily reflect the views of the United States Air Force or the United States Navy.}


%
\bibliographystyle{elsarticle-num}
\bibliography{kochdumper,cpsGroup}

\begin{thebibliography}{10}
\expandafter\ifx\csname url\endcsname\relax
  \def\url#1{\texttt{#1}}\fi
\expandafter\ifx\csname urlprefix\endcsname\relax\def\urlprefix{URL }\fi
\expandafter\ifx\csname href\endcsname\relax
  \def\href#1#2{#2} \def\path#1{#1}\fi

\bibitem{Plaku2016}
E.~Plaku, S.~Karaman, Motion planning with temporal-logic specifications:
  {P}rogress and challenges, AI Communications 29~(1) (2016) 151--162.

\bibitem{Xu2017b}
Z.~Xu, A.~Julius, J.~H. Chow, Energy storage controller synthesis for power
  systems with temporal logic specifications, IEEE Systems Journal 13~(1)
  (2017) 748--759.

\bibitem{Maierhofer2020}
S.~Maierhofer, A.-K. Rettinger, E.~C. Mayer, M.~Althoff, Formalization of
  interstate traffic rules in temporal logic, in: Proc. of the Intelligent
  Vehicles Symposium, 2020, pp. 752--759.

\bibitem{Krasowski2021}
H.~Krasowski, M.~Althoff, Temporal logic formalization of marine traffic rules,
  in: Proc. of the Intelligent Vehicles Symposium, 2021, pp. 186--192.

\bibitem{Wetzlinger2022}
M.~Wetzlinger, N.~Kochdumper, S.~Bak, M.~Althoff, Fully automated verification
  of linear systems using inner- and outer-approximations of reachable sets,
  Transactions on Automatic Control 68~(12) (2023) 7771--7786.

\bibitem{Maler2004}
O.~Maler, D.~Nickovic, Monitoring temporal properties of continuous signals,
  in: Proc. of the International Conference on Formal Modelling and Analysis of
  Timed Systems, 2004, pp. 152--166.

\bibitem{Donze2013}
A.~Donz{\'e}, T.~Ferrere, O.~Maler, Efficient robust monitoring for {STL}, in:
  Proc. of the International Conference on Computer Aided Verification, 2013,
  pp. 264--279.

\bibitem{Fisher2011}
M.~Fisher, An Introduction to Practical Formal Methods using Temporal Logic,
  John Wiley \& Sons, 2011.

\bibitem{Baier2008}
C.~Baier, J.-P. Katoen, Principles of Model Checking, MIT Press, 2008.

\bibitem{Gaiser2009}
A.~Gaiser, S.~Schwoon, Comparison of algorithms for checking emptiness on
  {B}{\"{u}}chi automata, in: Proc. of the Doctoral Workshop on Mathematical
  and Engineering Methods in Computer Science, 2009, {A}rticle 4.

\bibitem{Roehm2016b}
H.~Roehm, J.~Oehlerking, T.~Heinz, M.~Althoff, {STL} model checking of
  continuous and hybrid systems, in: Proc. of the International Symposium on
  Automated Technology for Verification and Analysis, 2016, pp. 412--427.

\bibitem{Tabuada2003}
P.~Tabuada, G.~J. Pappas, Model checking {LTL} over controllable linear systems
  is decidable, in: Proc. of the International Conference on Hybrid Systems:
  Computation and Control, 2003, pp. 498--513.

\bibitem{Yordanov2013}
B.~Yordanov, et~al., Formal analysis of piecewise affine systems through
  formula-guided refinement, Automatica 49 (2013) 261--266.

\bibitem{Gao2021}
Y.~Gao, et~al., Temporal logic trees for model checking and control synthesis
  of uncertain discrete-time systems, Transactions on Automatic Control 67~(10)
  (2021) 5071--5086.

\bibitem{Bresolin2013}
D.~Bresolin, {HyLTL}: {A} temporal logic for model checking hybrid systems, in:
  Proc. of the International Workshop on Hybrid Autonomous Systems, 2013, pp.
  73--84.

\bibitem{Frehse2018}
G.~Frehse, et~al., A toolchain for verifying safety properties of hybrid
  automata via pattern templates, in: Proc. of the American Control Conference,
  2018, pp. 2384--2391.

\bibitem{Pnueli1977}
A.~Pnueli, The temporal logic of programs, in: Proc. of the Annual Symposium on
  Foundations of Computer Science, 1977, pp. 46--57.

\bibitem{Lamport1993}
L.~Lamport, Hybrid systems in {$\text{TLA}^+$}, in: Proc. of the International
  Hybrid Systems Workshop, 1993, pp. 77--102.

\bibitem{Chen2020}
M.~Chen, Q.~Tam, S.~C. Livingston, M.~Pavone, Signal temporal logic meets
  reachability: Connections and applications, in: Proc. of the International
  Workshop on the Algorithmic Foundations of Robotics, 2020, pp. 581--601.

\bibitem{Yu2022}
G.~Yu, J.~Lee, K.~Bae, {STL}{\sc mc}: {R}obust {STL} model checking of hybrid
  systems using {SMT}, in: Proc. of the International Conference on Computer
  Aided Verification, 2022, pp. 524--537.

\bibitem{Lee2021}
J.~Lee, G.~Yu, K.~Bae, Efficient {SMT}-based model checking for signal temporal
  logic, in: Proc. of the International Conference on Automated Software
  Engineering, 2021, pp. 343--354.

\bibitem{Bae2019}
K.~Bae, J.~Lee, Bounded model checking of signal temporal logic properties
  using syntactic separation, Proceedings of the ACM on Programming Languages
  3~(POPL), {A}rticle 51 (2019).

\bibitem{Mitchell2005}
I.~M. Mitchell, A.~M. Bayen, C.~J. Tomlin, A time-dependent
  {Hamilton\textendash Jacobi} formulation of reachable sets for continuous
  dynamic games, Transactions on Automatic Control 50~(7) (2005) 947--957.

\bibitem{Althoff2010a}
M.~Althoff, Reachability analysis and its application to the safety assessment
  of autonomous cars, Ph.D. thesis, Technical University of Munich (2010).

\bibitem{Kochdumper2020c}
N.~Kochdumper, B.~Sch\"urmann, M.~Althoff, Utilizing dependencies to obtain
  subsets of reachable sets, in: Proc. of the International Conference on
  Hybrid Systems: Computation and Control, 2020, {A}rticle 1.

\bibitem{Kochdumper2022}
N.~Kochdumper, S.~Bak, Conformant synthesis for {K}oopman operator linearized
  control systems, in: Proc. of the International Conference on Decision and
  Control, 2022, pp. 7327--7332.

\bibitem{Farhadsefat2011}
R.~Farhadsefat, J.~Rohn, T.~Lotfi, Norms of interval matrices, Tech. rep.,
  Academy of Sciences of the Czech Republic, Institute of Computer Science
  (2011).

\bibitem{Lafferriere2001}
G.~Lafferriere, G.~J. Pappas, S.~Yovine, Symbolic reachability computation for
  families of linear vector fields, Symbolic Computation 32 (2001) 231--253.

\bibitem{Althoff2015a}
M.~Althoff, An introduction to {CORA} 2015, in: Proc. of the International
  Workshop on Applied Verification for Continuous and Hybrid Systems, 2015, pp.
  120--151.

\bibitem{ARCH22linear}
M.~Althoff, M.~Forets, C.~Schilling, M.~Wetzlinger, {ARCH-COMP22} category
  report: {C}ontinuous and hybrid systems with linear continuous dynamics, in:
  Proc. of the International Workshop on Applied Verification of Continuous and
  Hybrid Systems, 2022, pp. 58--85.

\bibitem{Koschi2020}
M.~Koschi, M.~Althoff, Set-based prediction of traffic participants considering
  occlusions and traffic rules, Transactions on Intelligent Vehicles 6~(2)
  (2020) 249--265.

\bibitem{Scott2016}
J.~K. Scott, D.~M. Raimondo, G.~R. Marseglia, R.~D. Braatz, Constrained
  zonotopes: A new tool for set-based estimation and fault detection,
  Automatica 69 (2016) 126--136.

\bibitem{Althoff2012a}
M.~Althoff, B.~H. Krogh, Avoiding geometric intersection operations in
  reachability analysis of hybrid systems, in: Proc. of the International
  Conference on Hybrid Systems: Computation and Control, 2012, pp. 45--54.

\bibitem{Wetzlinger2021}
M.~Wetzlinger, A.~Kulmburg, M.~Althoff, Adaptive parameter tuning for
  reachability analysis of nonlinear systems, in: Proc. of the International
  Conference on Hybrid Systems: Computation and Control, 2021, {A}rticle 16.

\bibitem{Kochdumper2019}
N.~Kochdumper, M.~Althoff, Sparse polynomial zonotopes: A novel set
  representation for reachability analysis, Transactions on Automatic Control
  66~(9) (2021) 4043--4058.

\bibitem{Chen2012}
X.~Chen, S.~Sankaranarayanan, E.~{\'A}brah{\'a}m, Taylor model flowpipe
  construction for non-linear hybrid systems, in: Proc. of the Real-Time
  Systems Symposium, 2012, pp. 183--192.

\bibitem{Jaulin2006}
L.~Jaulin, M.~Kieffer, O.~Didrit, Applied Interval Analysis, Springer Science
  \& Business Media, 2006.

\bibitem{Hamzeei2014}
M.~Hamzeei, J.~Luedtke, Linearization-based algorithms for mixed-integer
  nonlinear programs with convex continuous relaxation, Journal of Global
  Optimization 59~(2-3) (2014) 343--365.

\end{thebibliography}


\newpage
\appendix

\section{}
\label{sec:appendix}

We now demonstrate the conversion from signal temporal logic to reachset temporal logic according to \cite[Sec.~4]{Roehm2016b} for the exemplary STL formula
\begin{equation*}
	\varphi = \nextOp_{0.5} \, x_1 > 2 \vee \neg \finally_{[0,0.8]} \, x_2 \leq 3
\end{equation*}
and time step size $\Delta t = 0.5$. The first step is to convert the formula to negation normal form by moving all negations inward until they only appear in front of non-temporal expressions:
\begin{equation*}
	\varphi_{\text{nnf}} = \nextOp_{0.5} \, \underbrace{x_1 > 2}_{\pred_1} \vee \, \globally_{[0,0.8]} \, \underbrace{x_2 > 3}_{\pred_2}.
\end{equation*}
Next, we convert the formula to sampled-time STL with sampling period $\Delta t$. For this, we rewrite all times and time intervals for the temporal operators as integer multiples $i \, \Delta t$, $i \in \mathbb{N}_0$ of the time step size. If the start and end times are not divisible by the time step size we can either extend or shorten the corresponding time intervals in a sound matter. In our formula, for example, we can replace $\globally_{[0,0.8]} \, \pred_2$ by $\globally_{[0,1]} \, \pred_2$ since satisfaction of $\globally_{[0,1]} \, \pred_2$ implies satisfaction of $\globally_{[0,0.8]} \, \pred_2$. For time step size $\Delta t = 0.5$ we therefore obtain 
\begin{equation*}
	\varphi_{\text{st}_{\text{1}}} = \nextOp_{1} \, \pred_1 \vee \globally_{[0,2]} \, \pred_2,
\end{equation*}
which after applying the rewriting rules in \cite[Tab.~1]{Roehm2016b} results in the sampled-time STL formula
\begin{equation*}
	\varphi_{\text{st}_{\text{2}}} = \nextOp_{1} \pred_1 \vee \big( \hspace{-2pt} \nextOp_0 \hspace{-1pt}\pred_2 \wedge \globally_{[0,1]} \, \pred_2 \wedge \nextOp_1  \pred_2 \wedge \nextOp_1 \, \globally_{[0,1]} \, \pred_2 \wedge \nextOp_2 \, \pred_2 \big).
\end{equation*}
Afterward, we have to convert the formula to conjunctive normal form, which yields
\begin{equation*}
\begin{split}
	\varphi_{\text{cnf}} = \big( \hspace{-2pt} & \nextOp_{1} \hspace{-1pt}\pred_1 \vee \nextOp_0 \, \pred_2 \big) \wedge \big(\hspace{-2pt} \nextOp_{1} \hspace{-2pt}\pred_1 \vee \globally_{[0,1]} \, \pred_2 \big) \wedge \\
	& \nextOp_{1} \hspace{-1pt}(\pred_1 \vee \pred_2 ) \wedge \big(\hspace{-2pt} \nextOp_{1} \hspace{-1pt}\pred_1 \vee \nextOp_1 \, \globally_{[0,1]} \, \pred_2 \big) \wedge \big( \hspace{-2pt} \nextOp_{1}\hspace{-1pt} \pred_1 \vee \nextOp_2 \, \pred_2 \big).
\end{split}
\end{equation*}
Finally, we can apply the conversion to reachset temporal logic in \cite[Lemma~2]{Roehm2016b} to obtain
\begin{align*}
	\varphi_{\text{rtl}} = \big( & \hspace{-2pt} \nextOp_{1} \hspace{-2pt}\all \pred_1 \vee \nextOp_0 \all \pred_2 \big) \wedge \big( \hspace{-2pt} \nextOp_{1} \hspace{-2pt}\all \pred_1 \vee \nextOp_{0.5} \all \pred_2 \big) \wedge \\
	& \hspace{-2pt} \nextOp_{1} \hspace{-2pt}\all (\pred_1 \vee \pred_2 ) \wedge \big( \hspace{-2pt} \nextOp_{1} \hspace{-2pt}\all \pred_1 \vee \nextOp_{1.5} \all \pred_2 \big) \wedge \big( \hspace{-2pt} \nextOp_{1}\hspace{-2pt} \all \pred_1 \vee \nextOp_2 \all \pred_2 \big).
\end{align*}
Moreover, according to \cite[Sec.~5]{Roehm2016b} entailment can be equivalently formulated in terms of intersection checks with polytopes as in \eqref{eq:RTLfinal}:
\begin{align*}
	\mathcal{R} \vDash \varphi_{\text{rtl}} ~~ \Leftrightarrow ~~ & \big( \mathcal{R}(t_1)  \cap \mathcal{P}_1 = \emptyset \vee \mathcal{R}(t_0) \cap \mathcal{P}_2 = \emptyset \big) \wedge \\
	& \big( \mathcal{R}(t_1) \cap \mathcal{P}_1 = \emptyset \vee \mathcal{R}(\tau_0) \cap \mathcal{P}_2 = \emptyset \big) \wedge \\
	& ~ \mathcal{R}(t_1) \cap \mathcal{P}_{12} = \emptyset \wedge \\
	&    \big( \mathcal{R}(t_1) \cap \mathcal{P}_1 = \emptyset \vee \mathcal{R}(\tau_1) \cap \mathcal{P}_2 = \emptyset \big) \wedge \\
	&  \big( \mathcal{R}(t_1) \cap \mathcal{P}_1 = \emptyset \vee \mathcal{R}(t_2) \cap \mathcal{P}_2 = \emptyset \big),
\end{align*}
where the polytopes $\mathcal{P}_1$, $\mathcal{P}_2$, and $\mathcal{P}_{12}$ corresponding to the predicates $\pred_1$, $\pred_2$, and $\pred_{1} \vee \pred_{2}$ are defined as
\begin{equation*}
	\mathcal{P}_1 = \big \langle [1~0],2 \big \rangle_P, ~~ \mathcal{P}_2 = \big \langle [0~1],3 \big \rangle_P, ~~ \mathcal{P}_{12} = \bigg \langle \begin{bmatrix} 1~ & ~0 \\ 0~ & ~1 \end{bmatrix},\begin{bmatrix} 2 \\ 3 \end{bmatrix} \bigg \rangle_P.
\end{equation*}
The resulting formula can therefore be directly evaluated on the reach sequence $\mathcal{R}(t_0),\mathcal{R}(\tau_0),\mathcal{R}(t_1),\mathcal{R}(\tau_1),\mathcal{R}(t_2)$.

\section{}
\label{sec:appendixB}

We now prove that the enclosure computed with the reachability algorithm in Alg.~\ref{alg:reach} converges to the exact reachable set for $\Delta t \to 0$. In summary, Alg.~\ref{alg:reach} contains three sources of over-approximation: 
\begin{enumerate}
	\item The enclosure of the difference between the reachable set due to constant and time-varying inputs $\mathcal{D}$.
	\item The curvature enclosure $\mathcal{C}_i$.
	\item The over-approximation in the zonotope enclosure of the convex hull \\ $conv(\mathcal{H}(t_i),\mathcal{H}(t_{i+1}))$ computed using \eqref{eq:zonoConvHull}. 
\end{enumerate}
We therefore have to show that all these over-approximation errors converge to zero for $\Delta t \to 0$. We begin with the difference between constant and time-varying inputs:

\begin{proposition}
    The enclosure of the difference between the reachable set due to constant and time-varying inputs $\mathcal{D}$ according to \eqref{eq:inputDiff} satisfies
	\begin{equation*}
		\lim_{\Delta t \to 0} \mathcal{D} = \mathbf{0}.
	\end{equation*}	
\end{proposition}
\begin{proof}
	With the formula for $\mathcal{D}$ in \eqref{eq:inputDiff}, we obtain
	\begin{align*}
		\lim_{\Delta t \to 0} \mathcal{D} & \overset{\eqref{eq:inputDiff}}{=}	\lim_{\Delta t \to 0} \bigg( \sum_{j=1}^\kappa \frac{A^j \Delta t^{j+1}}{(j+1)!} \bigg) \mathcal{U}_0 \oplus \bigoplus_{j=1}^\kappa \frac{A^j \Delta t^{j+1}}{(j+1)!}\, \mathcal{U}_0 \oplus 2 \,  \Delta t \, \mathbfcal{E} \, \mathcal{U}_0 \\
		& ~= \bigg( \sum_{j=1}^\kappa \frac{A^j \, 0}{(j+1)!} \bigg) \mathcal{U}_0 \oplus \bigoplus_{j=1}^\kappa \frac{A^j \, 0}{(j+1)!}\, \mathcal{U}_0 \oplus 2 \,  0 \, \mathbfcal{E} \, \mathcal{U}_0 = \mathbf{0},
	\end{align*}
	which concludes the proof.
\end{proof}

Next, we consider the curvature enclosure:
\begin{proposition}
	The curvature enclosure $\mathcal{C}_i$ according to \eqref{eq:Cu} satisfies
	\begin{equation*}
		\lim_{\Delta t \to 0} \mathcal{C}_i = \mathbf{0}.
	\end{equation*}
\end{proposition} 
\begin{proof}
 	According to \cite[Lemma~3]{Wetzlinger2022}, it holds that
 	\begin{equation}
 		\lim_{\Delta t \to 0} \mathbfcal{F} = \mathbf{0} ~~~~ \text{and} ~~~~ \lim_{\Delta t \to 0} \mathbfcal{G} = \mathbf{0},
 		\label{eq:proofC}
 	\end{equation}
 	with $\mathbfcal{F}$ and $\mathbfcal{G}$ defined as in \eqref{eq:curv}. With the formula for $\mathcal{C}_i$ in \eqref{eq:Cu} we therefore obtain
	\begin{align*}
		\lim_{\Delta t \to 0} \mathcal{C}_i & \overset{\eqref{eq:Cu}}{=} \lim_{\Delta t \to 0} (e^{A \Delta t})^{i} \big( \mathbfcal{F} \, \mathcal{X}_0 \oplus \mathbfcal{G} \, \widetilde{u} \big) = \lim_{\Delta t \to 0} \underbrace{(e^{A \, 0})^{i}}_{=I_n} \big( \mathbfcal{F} \, \mathcal{X}_0 \oplus \mathbfcal{G} \, \widetilde{u} \big) \\
		& ~\, = \lim_{\Delta t \to 0}  \mathbfcal{F} \, \mathcal{X}_0 \oplus \mathbfcal{G} \, \widetilde{u} \overset{\eqref{eq:proofC}}{=} \mathbf{0},
	\end{align*}
	which concludes the proof.
\end{proof}

Finally, we examine the over-approximation introduced by the enclosure of the convex hull: 

\begin{proposition}
	Given two zonotopes $\mathcal{Z}_1,\mathcal{Z}_2 \subset \mathbb{R}^n$, let $\text{\normalfont\texttt{diff}}(\mathcal{Z}_1,\mathcal{Z}_2)$ denote the difference between the exact convex hull $conv(\mathcal{Z}_1,\mathcal{Z}_2)$ as defined in \eqref{eq:defConvHull} and the enclosure computed according to \eqref{eq:zonoConvHull}. The homogeneous solution $\mathcal{H}(t_{i})$ in \eqref{eq:propScheme} satisfies:
	\begin{equation*}
		\lim_{\Delta t \to 0} \text{\normalfont\texttt{diff}}\big(\mathcal{H}(t_i),\mathcal{H}(t_{i+1})\big) = \mathbf{0}.
	\end{equation*}	
\end{proposition}
\begin{proof}
	For the limit case $\Delta t \to 0$, we obtain the following for the convex hull $conv \big(\mathcal{H}(t_i),\mathcal{H}(t_{i+1}) \big)$:
	\begin{equation}
	\begin{split}
		\lim_{\Delta t \to 0} conv \big(\mathcal{H}(t_i),\mathcal{H}(t_{i+1}) \big) &\overset{\eqref{eq:propScheme}}{=} \lim_{\Delta t \to 0} conv \big(\mathcal{H}(t_{i}), e^{A \Delta t} \, \mathcal{H}(t_i) \oplus T \, \widetilde{u} \big) \\
		&~= \lim_{\Delta t \to 0} conv \big(\mathcal{H}(t_{i}), e^{A \Delta t} \, \mathcal{H}(t_i) \oplus A^{-1} (e^{A \Delta t} - I_n) \, \widetilde{u} \big) \\
		&~=  conv \big(\mathcal{H}(t_{i}), \underbrace{e^{A \, 0}}_{=I_n} \, \mathcal{H}(t_i) \oplus A^{-1} (\underbrace{e^{A \, 0}}_{=I_n} - I_n) \, \widetilde{u} \big) \\
		&~= conv \big(\mathcal{H}(t_{i}),\mathcal{H}(t_i)\big).
	\end{split}
	\label{eq:proofConv1}
	\end{equation}
	Moreover, given a set $\mathcal{S} \subset \Rn$, the exact convex hull as defined in \eqref{eq:defConvHull} satisfies 
	\begin{equation}
	conv(\mathcal{S},\mathcal{S}) = \bigg \{\sum_{i=1}^{n+1} \lambda_i\,s_i~\bigg|~s_i \in \underbrace{\mathcal{S} \cup \mathcal{S}}_{=\mathcal{S}},~\lambda_i \geq 0,~\sum_{i=1}^{n+1} \lambda_i = 1 \bigg\} = \mathcal{S}.
	\label{eq:proofConv2}
	\end{equation}
	In addition, given a zonotope $\mathcal{Z} = \langle c,G \rangle_Z$, the enclosure of the convex hull according to \eqref{eq:zonoConvHull} satisfies
	\begin{equation}
	\begin{split}
		conv\big(\mathcal{Z},\mathcal{Z} \big) & \overset{\eqref{eq:zonoConvHull}}{=} \big\langle 0.5(c+c),\big[0.5(G+G)~0.5(G-G)~0.5(c-c) \big]\big\rangle_Z  \\[5pt]
		& = \big \langle c, \big[G~ \mathbf{0} ~ \mathbf{0}\big] \big \rangle_Z = \mathcal{Z}.
	\end{split}
	\label{eq:proofConv}
	\end{equation}
	For $conv \big(\mathcal{H}(t_{i}),\mathcal{H}(t_i)\big)$ in \eqref{eq:proofConv1}, both, the exact convex hull as defined in \eqref{eq:defConvHull} as well as the zonotope enclosure computed using \eqref{eq:zonoConvHull} yield the result $\mathcal{H}(t_i)$ according to \eqref{eq:proofConv2} and \eqref{eq:proofConv}, which proves that the difference between the exact convex hull and the enclosure computed according to \eqref{eq:zonoConvHull} converges to 0 for $\Delta t \to 0$.
\end{proof}


\end{document}